\documentclass[pra,onecolumn,longbibliography,superscriptaddress]{revtex4-2}

\usepackage[dvips]{graphicx} 
\usepackage{amsfonts,amssymb,amscd,amsmath,amsthm}
\usepackage{textcomp}
\usepackage{enumerate}
\usepackage{epsfig}
\usepackage{subfigure}
\usepackage{xcolor}
\usepackage[outdir=./]{epstopdf}
\usepackage[colorlinks = true]{hyperref}
\usepackage{physics}
\usepackage{appendix}
\usepackage{comment, apptools, thmtools, thm-restate}
\usepackage{diagbox}
\usepackage[most]{tcolorbox}

\usepackage{tikz}
\usepackage[]{qcircuit}
\usetikzlibrary{arrows}
\usetikzlibrary{shapes,fadings,snakes}
\usetikzlibrary{decorations.pathmorphing,patterns}
\usetikzlibrary{calc}
\usetikzlibrary{positioning}
\usepackage{etex}
\usepackage{pgfplots}
\usepgfplotslibrary{ternary}
\pgfplotsset{compat=newest}

\newtheorem{theorem}{Theorem}

\newtheorem{corollary}{Corollary}

\newcommand{\lket}[1]{\vert #1 \rangle\!\rangle}
\newcommand{\lbra}[1]{\langle\!\langle #1 \vert}
\newcommand{\lbraket}[2]{\langle\!\langle #1 \vert #2 \rangle\!\rangle}
\newcommand{\lketbra}[2]{\vert #1 \rangle\!\rangle\langle\!\langle #2 \vert}

\newtcolorbox[auto counter]{mybox}[2][]{
	enhanced,
	breakable,
	colback=blue!5!white,
	colframe=blue!75!black,
	fonttitle=\bfseries,
	title=Box \thetcbcounter: #2,#1
}

\graphicspath{{./figure/}}

\begin{document}

\title{Robust Estimation of Nonlinear Properties of Quantum Processes}

\begin{abstract}
Accurate and robust estimation of quantum process properties is crucial for quantum information processing and quantum many-body physics. Combining classical shadow tomography and randomized benchmarking, Helsen et al.~introduced a method to estimate the linear properties of quantum processes. In this work, we focus on the estimation protocols of nonlinear process properties that are robust to state preparation and measurement errors. We introduce two protocols, both utilizing random gate sequences but employing different post-processing methods, which make them suitable for measuring different nonlinear properties. The first protocol offers a robust and sound method to estimate the out-of-time-ordered correlation, as demonstrated numerically in an Ising model. The second protocol estimates unitarity, effectively characterizing the incoherence of quantum channels. We expect the two protocols to be useful tools for exploring quantum many-body physics and characterizing quantum processes.
\end{abstract}
 \date{\today}
 
\author{Yuqing Wang}

\affiliation{Center for Quantum Information, Institute for Interdisciplinary Information Sciences, Tsinghua University, Beijing 100084, P.~R.~China}

\author{Guoding Liu}
\email{lgd22@mails.tsinghua.edu.cn}

\affiliation{Center for Quantum Information, Institute for Interdisciplinary Information Sciences, Tsinghua University, Beijing 100084, P.~R.~China}

\author{Zhenhuan Liu}
\email{liu-zh20@mails.tsinghua.edu.cn}

\affiliation{Center for Quantum Information, Institute for Interdisciplinary Information Sciences, Tsinghua University, Beijing 100084, P.~R.~China}

\author{Yifan Tang}
\affiliation{Department of Mathematics and Computer Science, Freie Universit\"{a}t Berlin, 14195 Berlin, Germany}

\author{Xiongfeng Ma}
\email{xma@tsinghua.edu.cn}

\affiliation{Center for Quantum Information, Institute for Interdisciplinary Information Sciences, Tsinghua University, Beijing 100084, P.~R.~China}
\author{Hao Dai}
\email{dhao@bimsa.cn}
\affiliation{Yanqi Lake Beijing Institute of Mathematical Sciences and Applications (BIMSA), Huairou District, Beijing 101408, P.~R.~China}

\maketitle
\tableofcontents
\clearpage

\section{Introduction}
With the development of quantum technology, the ability to control large quantum systems enables us to simulate extensive quantum many-body systems and investigate associated phenomena. Quantum processes, central to quantum physics, are mathematically described as quantum channels, which are completely positive and trace-preserving maps. All properties of interest within quantum processes are functions of the quantum channel, such as the scrambling strength~\cite{Swingle2018Unscrambling}, the unitarity~\cite{Wallman2015coherence}, and the similarity between two quantum processes~\cite{Belavkin2005distance}. While quantum process tomography~\cite{Chuang1997tomo} offers a straightforward approach to estimating all these properties, its complexity grows exponentially with the number of qubits, rendering it impractical even for small-scale quantum systems.

Fortunately, full knowledge of a quantum channel is not always necessary to estimate some specific properties. Estimating partial knowledge of quantum channels can significantly reduce sample complexity compared to full tomography. This concept is analogous to progress in quantum state learning~\cite{Aaronson2018shadow}, such as classical shadow tomography~\cite{Huang2020shadow} and its variations~\cite{chen2021robust,hadfield2022measurements,huang2021efficient,Elben2023toolbox}, which utilize random measurements for efficient estimation of state properties. The process of classical shadow involves applying random unitary evolutions to the target state followed by computational basis measurements, allowing for the estimation of multiple state properties simultaneously. This approach has been adapted for channel property estimation through the Choi–Jamiołkowski isomorphism~\cite{levy2021classical}. Specifically, one inputs random states to the quantum channel and performs randomized measurements on output states, equivalent to performing randomized measurements on the Choi state of the target channel. Like state shadow tomography, this technique allows the simultaneous estimation of multiple quantum channel properties.

While the aforementioned method can efficiently estimate certain channel properties without exponential sample complexities, it lacks robustness against state preparation and measurement (SPAM) errors. In many practical systems, SPAM errors, particularly measurement errors, can be as significant as or even surpass quantum gate errors~\cite{StrongQuantumComputation,Supremacy}. Therefore, accurately estimating channel properties necessitates mitigating the influence of SPAM errors. Randomized benchmarking serves as a widely adopted protocol for this purpose, allowing property estimation of quantum channels while reducing the impact of SPAM errors~\cite{helsen2022framework}. However, conventional randomized benchmarking protocols are limited in measurable properties, primarily restricted to properties like average fidelity~\cite{knill2008rb}.

Utilizing group twirling and fitting techniques, Helsen et al.~have combined randomized benchmarking with classical shadow to estimate arbitrary linear properties of quantum channels, robust against SPAM errors~\cite{helsen2021estimating}. However, to fully explore quantum phenomena and characterize quantum channels, linear properties alone are insufficient. Many critical nonlinear properties exist, such as out-of-time-ordered correlation (OTOC)~\cite{Swingle2018OTOC}, a measure of information scrambling. OTOC is critical in both quantum many-body physics \cite{FAN2017otocmbl,Chen2023LRB} and quantum information \cite{shen2020scrambling,Leone2021quantumchaosis}. Yet, current estimation methods for OTOC lack robustness against SPAM errors \cite{Vermersch_2019,Sundar2022otoc,green2022otoc,wang2021verifying,yoshida2019disentangling,xiao2021scrambling,Garcia2021scrambling}, posing a challenge in observing information scrambling. In addition, other nonlinear properties such as unitarity and magic, which characterize the incoherence and non-stabilizerness of a quantum channel, are also important and widely discussed in quantum information field \cite{Wallman2015coherence,Nonstabilizerness2023Leone}.

Building upon Helsen et al.'s framework~\cite{helsen2021estimating}, we introduce two protocols to estimate nonlinear properties of quantum channels with robustness against SPAM errors, as shown in Fig.~\ref{fig:overview}. Both protocols implement random gate sequences sampled from a group and perform a positive operator-valued measure (POVM) to collect shadow data at the first step. Then, they differ in classical data post-processing. The first protocol utilizes correlations between measurement data from independently chosen gate sequences, while the second harnesses correlations from identical gate sequences. The expectation of the correlation is a multiple exponential decay function with the circuit depth. Through exponential fitting, one can get nonlinear properties of quantum channels excluding the SPAM error. Different classical postprocessing procedures make the two protocols suitable for evaluating different channel properties and exploring different phenomena within quantum many-body systems. As an application, we employ the first protocol to measure OTOC and theoretically analyze its sample complexity. In a long-range interaction Ising model, we numerically demonstrate the effectiveness of the protocol and resilience to SPAM errors. Moreover, we explore the potential of the second protocol in estimating unitarity~\cite{Wallman2015coherence}. We analyze the type of channel properties that can be measured with each protocol when the matrix exponential fitting is allowed. We find that, when the gate sequences are sampled from a unitary two-design group, the measurable quantities of the first protocol cover those of the second one.

\begin{figure}
    \centering
    \includegraphics[width=18cm]{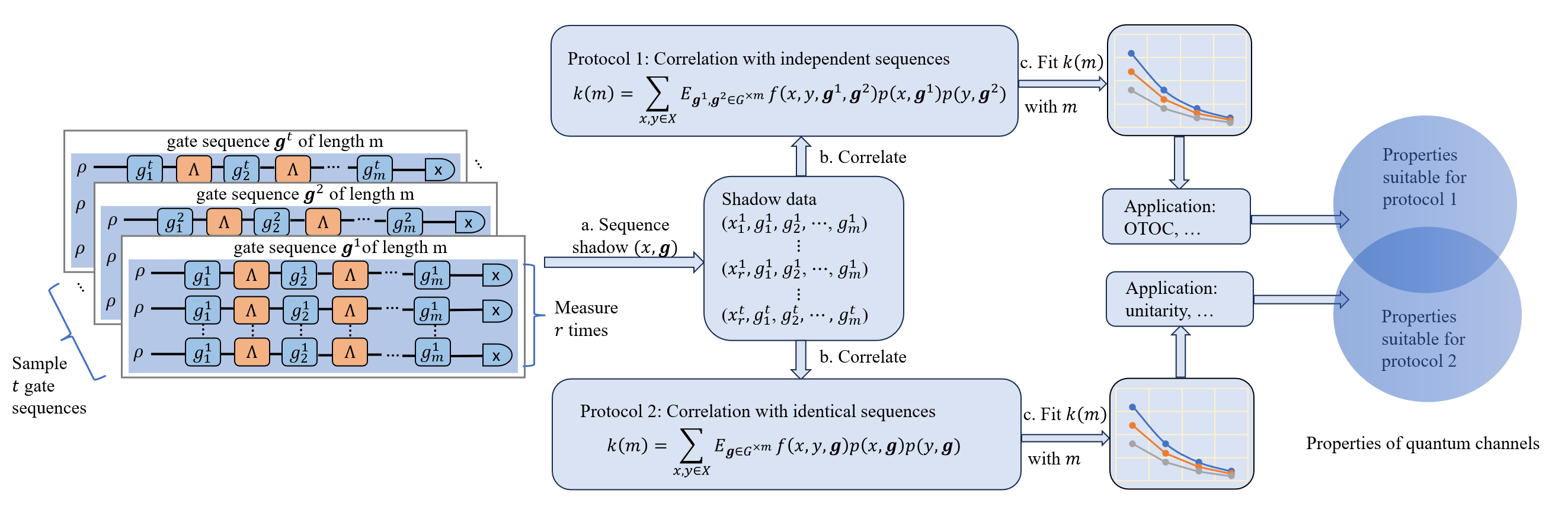}
    \caption{
    The workflow of the nonlinear channel properties estimation via the generalized uniformly independent random sequence protocol introduced in Sec~\ref{sc:ncshadow} and Sec.~\ref{sc:ncrandom}. In step a, we obtain the shadow sequence data $(x,\mathbf{g})$ by implementing random gate sequences and POVM measurements. For each sequence length $m$, we measure the same gate sequences for $r$ times and iterate the process with $t$ different gate sequences. In step b, we calculate the correlation function $k(m)$ using the shadow sequence data. The first protocol utilizes data from independent sequences, while the second protocol involves data from identical sequences. The entire procedure is repeated for several different sequence lengths $m$. In step c, we perform a fitting of the correlation function $k(m)$ against the sequence length $m$ to extract desired nonlinear channel properties.
    With different group $G$, the two protocols may be suitable for evaluating different channel properties. In particular, if $G$ is the global $n$-qubit Clifford group, protocol 1 can estimate all the properties that protocol 2 can obtain, as discussed in Sec~\ref{ssc:measure}.}
    \label{fig:overview}
\end{figure}

This work is organized as follows. In Section~\ref{sc:pre}, we introduce basic notations and the necessary preliminaries for this work. In Section~\ref{sc:nonlinear}, we present two protocols to estimate nonlinear channel properties, provide two illustrative examples, and discuss the measurable properties of these two protocols. In Section~\ref{sc:simulation}, we showcase the results of our simulations pertaining to OTOC estimation. We conclude in Section~\ref{sc:conclusion}.

\section{Preliminaries}\label{sc:pre}
In this section, we define key notations and review essential concepts in the Pauli-Liouville representation, the Clifford group, and the uniformly independent random sequence (UIRS) shadow \cite{helsen2021estimating}.

\subsection{Pauli-Liouville representation}
For an $n$-qubit system, $\mathbb{C}^{d}$ with the dimension $d=2^n$, the normalized Pauli group is defined as 
\begin{equation}
	\mathbb{P}_n=\left\{\frac{\sigma}{\sqrt{d}} \Bigm| \sigma\in\{I,\sigma_x,\sigma_y,\sigma_z\}^{\otimes n} \right\}. 
\end{equation}
The elements of this group, $I,\sigma_x,\sigma_y$, and $\sigma_z$, represent the identity operator, and Pauli $X$, $Y$, and $Z$ operators, respectively, defined as,
\begin{equation}
	I = \begin{pmatrix}
		1 & 0\\
		0 & 1
	\end{pmatrix},
	\sigma_x = \begin{pmatrix}
		0 & 1\\
		1 & 0
	\end{pmatrix},
	\sigma_y = \begin{pmatrix}
		0 & -i\\
		i & 0
	\end{pmatrix},
	\sigma_z = \begin{pmatrix}
		1 & 0\\
		0 & -1
	\end{pmatrix}.
\end{equation}
In certain contexts, the identity operator is excluded, leading to the set $\mathbb{P}_n^0=\mathbb{P}_n\backslash \frac{\mathbf{1}}{\sqrt{d}}$, where we define $\mathbf{1}=I^{\otimes n}$. The Pauli group forms an orthonormal basis in the linear operator space $L(\mathbb{C}^{d})$, equipped with the inner product $\langle A, B \rangle=\Tr(A^{\dagger}B)$. Consequently, any operator $O$ can be expressed as a sum over this basis,
\begin{equation}
	O=\sum_{\sigma_i\in\mathbb{P}_n }\langle \sigma_i, O \rangle \sigma_i.
\end{equation}

The Pauli-Liouville representation vectorizes the operator space, defined as the linear map $\lket{\cdot}: L(\mathbb{C}^{d}) \rightarrow \mathbb{C}^{d^2}$, by assigning $\lket{\sigma_i} = e_i$, where $\sigma_i \in \mathbb{P}_n$ and $\{e_i\}$ forms an orthonormal basis in $\mathbb{C}^{d^2}$. Furthermore, the inner product in the Hilbert space $\mathbb{C}^{d^2}$ is represented as $\lbraket{A}{B} = \langle A, B \rangle = \Tr(A^{\dagger}B)$. By linearity, the Pauli-Liouville representation of an operator $O$ is expressed as,
\begin{equation}
	\lket{O} = \sum_{\sigma_i \in \mathbb{P}_n} \langle \sigma_i, O \rangle \lket{\sigma_i}.
\end{equation}

For a quantum channel $\mathcal{E}: L(\mathbb{C}^{d}) \rightarrow L(\mathbb{C}^{d})$, it can be represented as a matrix in the Pauli-Liouville representation. The elements of this matrix are given by,
\begin{equation}
	\mathcal{E}_{i,j}=\lbra{\sigma_i}\mathcal{E}\lket{\sigma_j}=\lbraket{\sigma_i}{\mathcal{E}(\sigma_j)}=\Tr(\sigma_i \mathcal{E}(\sigma_j)).  
\end{equation}
In this representation, the composition of channels corresponds to matrix multiplication. Specifically, for two quantum channels $\mathcal{E}_1$ and $\mathcal{E}_2$ acting on an operator $O$, the composition is expressed as,
\begin{equation}
	\lket{\mathcal{E}_1\circ \mathcal{E}_2(O)}=\mathcal{E}_1\mathcal{E}_2\lket{O}.
\end{equation}

\subsection{Clifford group}\label{ssc:Clifford}
The Clifford group $\mathcal{C}_n$ comprises $n$-qubit unitary operators that normalize the Pauli group $\mathbb{P}_n$. An operator $U$ is an element of the Clifford group if and only if for every Pauli operator $P \in \mathbb{P}_n$, the conjugated operator $UPU^\dagger$ is also a Pauli operator, up to a global phase. Mathematically, this is expressed as: for each $P \in \mathbb{P}_n$ and $U\in\mathcal{C}_n$, there exists $P' \in \mathbb{P}_n$ and a global phase $e^{i\theta\pi}$ such that $ UPU^{\dagger} = e^{i\theta\pi} P'$.

In the Pauli-Liouville representation, the Clifford group decomposes into two nonequivalent irreducible representations, $\forall g \in \mathcal{C}_n$,
\begin{equation}
	\omega(g)=\tau_{tr}(g)\oplus \tau_{ad}(g),
\end{equation}
where $\tau_{tr}$ is the trivial representation supported on the normalized identity matrix, $\lket{\mathbf{1}/ \sqrt{d}}$, and $\tau_{ad}$ is the adjoint representation supported on the traceless matrices. The projector onto the representation space of $\tau_{tr}$ is $P_{tr} = \lketbra{\mathbf{1}}{\mathbf{1}}/d$, while the projector onto the irreducible representation space of $\tau_{ad}$ is given by $P_{ad} = \sum_{\sigma \in \mathbb{P}_n^0}\lketbra{\sigma}{\sigma}$. 

\subsection{Uniformly independent random sequence shadow}
The uniformly independent random sequence shadow (UIRS)~\cite{helsen2021estimating} can estimate the linear properties of noisy gate sets by combining the methods of randomized benchmarking and classical shadow tomography. The protocol starts with the collection of shadow sequence data obtained from a random sequence of gates, followed by a POVM measurement. The subsequent step is computing the correlation function from the shadow data, which encodes the desired properties of the gate set. The final step entails estimating the gate set properties through a fitting process \cite{helsen2021estimating}. Notably, the UIRS protocol exhibits robustness against SPAM errors, similar to the conventional randomized benchmarking protocol.

Now, we proceed to formalize the UIRS protocol mathematically, following the procedures depicted in Fig~\ref{fig:overview}. Let $\mathbf{g} =(g_1,\cdots,g_m)$ denote a random sequence, comprising $m$ gates where each gate $g_i$ is uniformly and independently selected from a unitary gate-set $\mathbb{G}$. 
Typically, $\mathbb{G}$ is chosen as a group, and we adhere to this convention in our work. Given an input state, $\rho$, and a POVM with a finite-outcomes set $\mathcal{X}$, $\{E_{x}\}_{x\in\mathcal{X}}$, each iteration of the protocol yields a piece of shadow data, denoted as $(x,\mathbf{g})$. This data contains the measurement outcome $x$ and the random gate sequence $\mathbf{g}$. The probability of acquiring this specific data is given by:
\begin{equation}
p(x,\mathbf{g})=\lbra{E_x}\mathcal{E}(\mathbf{g})\lket{\rho},
\end{equation}
where $\mathcal{E}(\mathbf{g})$ denotes the noisy implementation of the random sequence $\mathbf{g}$.

In experiments, the presence of noise is inevitable. We model the noisy initial state and measurement as $\tilde{\rho}$ and $\{\tilde{E_x}\}$, respectively. Regarding the noisy implementation of the random sequence, we assume gate-independent noise, a special type of Markovian noise. Specifically, for each gate $g_i$ in the random sequence, the ideal unitary implementation is given by
\begin{equation}
\omega(g_i)\lket{\rho}=\lket{g_i(\rho)},
\end{equation}
where $\omega$ is the Liouville representation of the group $\mathbb{G}$ and $g_i(\rho)=U_{g_i}\rho U_{g_i}^{\dagger}$. 
The actual implementation of the noisy gate is given by 
\begin{equation}\label{eq:noisygi}
\phi(g_i) = \Lambda_L \omega(g_i) \Lambda_R,
\end{equation}
where $\Lambda_L$ and $\Lambda_R$ are noise channels independent of the gate choice. Then, the channel $\Lambda = \Lambda_R \Lambda_L$ denotes the noise occurring between gates. Consequently, the overall implemented channel for the random sequence is
\begin{equation}
	\mathcal{E}(\mathbf{g})=\prod_{i=1}^{m}\phi(g_i).
\end{equation}
Under these assumptions and considering errors in state and measurement, the probability of obtaining data $(x,\mathbf{g})$ is given by
\begin{equation}
	p(x,\mathbf{g})=\lbra{\tilde{E}_x}\prod_{i=1}^{m}\phi(g_i)\lket{\tilde{\rho}}.
\end{equation}

After multiple independent experimental rounds, we acquire a collection of gate-set shadows $\{(x_i,\mathbf{g}_i)\}_{i=1}^S$. From the shadow data, we aim to extract meaningful information, typically in the form of an expectation value of a sequence correlation function. We consider a correlation function $f(x,\mathbf{g},m): \mathcal{X} \times \mathbb{G}^{\times m} \times \mathbb{Z^+} \rightarrow \mathbb{C}$, with a specific form,
\begin{equation}\label{eq:UIRScorr}
f_A(x,\mathbf g, m)=\Tr(B_x\tau(g_m)\prod_{i=1}^{m-1}A\tau(g_i)),  
\end{equation}
where $\tau$ represents an irreducible representation of group $\mathbb{G}$ and $A$ and $B_x$ denote preselected operators supported on the representation space associated with $\tau$. The expectation of this correlation function across all possible choices of random gate sequences is given by
\begin{equation}
k_A(m)=\mathop{\mathbb E}\limits_{\mathbf g\in \mathbb{G}^{\times m}}\sum_{x\in\mathcal X}f_A(x,\mathbf g, m)p(x,\mathbf g).
\end{equation}
The quantity $k_A(m)$ is always referred to as the correlator with sequence correlation function $f_A$. It has been proven that $k_A(m)$ has an exponential behavior in $m$~\cite{helsen2021estimating},
\begin{equation}
k_A(m)=\Tr(\Theta(\{E_x\}_x,\rho)[\Phi(A,\Lambda)]^{m-1}),  
\end{equation}
where $\Theta(\{E_x\}_x,\rho)$ is a matrix dependent solely on SPAM and $B_x$; $\Phi(A,\Lambda)$ is a matrix dependent on $\mathbb{G}$, $A$, and $\Lambda$. If the Liouville representation of $\mathbb{G}$ has decomposition, $\omega(g)=\tau(g)^{n_{\tau}}\oplus\omega'(g)$, where $\omega'(g)$ contains no copy of $\tau(g)$, then $\Phi(A,\Lambda)$ reduces to
\begin{equation}
\Phi(A,\Lambda)_{i,j}=\frac{1}{|P_j|}\Tr(P_iAP_j\Lambda),  
\end{equation}
where $P_i$ is the projector onto the $i$-th copy of $\tau(g)$ inside the representation of $\omega(g)$ and $|P_j| = \rank(P_j)$ represents the dimension of $\tau(g)$. With the exponential fitting, one can obtain $\Phi(A, \Lambda)$ while excluding the influence of the SPAM, $\rho$ and $E_x$. Each matrix element of $\Phi(A, \Lambda)$, $\Tr(P_iAP_j\Lambda)$, is a linear function of $\Lambda$ controlled by $A$. We can thus estimate linear properties of $\Lambda$.

In reality, we cannot get the correlator $k_A$ perfectly due to the finite sampling, but we can construct an estimator, $\hat{k}_A(m)$, from the shadow data such that $\hat{k}_A(m)$ converges to $k_A(m)$ when the number of experiment rounds is sufficiently large. A standard estimator is given by
\begin{equation}
\hat{k}_A(m)=\frac{1}{S} \sum_{i=1}^{S}f_A(x_i,\mathbf{g}_i, m).
\end{equation}
When $S\rightarrow \infty$, this estimator will converge to $k_A(m)$ and allow one to get $\Phi(A, \Lambda)$.

In summary, the UIRS protocol serves as a valuable tool for estimating linear properties of noise channels associated with gate sets. Alternatively, when considering a noiseless gate-set $\mathbb{G}$ and artificially inserting the noise channel $\Lambda$ between two consecutive random gates $g_i$ and $g_{i+1}$, the UIRS protocol can be interpreted as a method for extracting properties of any given channel. Typically, $\Lambda$ can be chosen as a unitary evolution, which is of particular interest in studies of quantum many-body systems. In the following section, we will extend the UIRS protocol to the estimation of nonlinear properties of quantum channels, focusing on evaluating the OTOC of unitary quantum evolution and the unitarity of the quantum channels.

\section{Nonlinear channel properties estimation via generalized UIRS}\label{sc:nonlinear}
In this section, we present the generalized UIRS protocol to estimate nonlinear channel properties. Our protocol shares the quantum procedure with the UIRS protocol, that is, applying the random gate sequences and measurements to obtain the shadow data $(x,\mathbf{g})$. The difference lies in the classical postprocessing stage. In the UIRS protocol, the correlation function only involves one piece of shadow data like Eq.~\eqref{eq:UIRScorr}. In our protocol, we evaluate the correlation function by incorporating two sets of shadow data instead of one. The gate sequences in two pieces of shadow data can be chosen as independent or identical, corresponding to two different protocols. We study the two cases separately and draw a comparative analysis.

\subsection{Nonlinear correlation via independent sequences}\label{sc:ncshadow}
We first introduce the generalized UIRS protocol utilizing two pieces of shadow data derived from independent gate sequences. To illustrate, we provide a specific example demonstrating the effective estimation of the OTOC within this protocol.

After obtaining the shadow data, we can construct a sequence correlation function that encodes the information of the desired quantity. In this context, we consider a second-order sequence correlation function defined as
\begin{equation}\label{eq:independcorr}
f(x,y,\mathbf{g}^1,\mathbf{g}^2,m):\mathcal{X}\times \mathcal{X}\times \mathbb{G}^{\times m}\times \mathbb{G}^{\times m}\times\mathbb{Z^+} \rightarrow \mathbb{C},
\end{equation}
which depends on two independent random sequences $\mathbf{g}^1$ and $\mathbf{g}^2$ and their measurement outcomes $x$ and $y$. 
Notice that the function is analogous to the observable with respect to two copies of the state in the case of state shadow tomography. The expectation of the correlation function is taken over the group $\mathbb{G}^{\times m}\times \mathbb{G}^{\times m}$ and the measurement outcomes:
\begin{equation}\label{eq:expcor} 
k_f(m)=\sum_{x,y\in\mathcal{X}} \mathbb{E}_{\mathbf{g}^1,\mathbf{g}^2\in \mathbb{G}^{\times m}} f(x,y,\mathbf{g}^1,\mathbf{g}^2,m)p(x,\mathbf{g}^1)p(y,\mathbf{g}^2).
\end{equation}
With the sequence shadow data $\{(x_i,\mathbf{g}^i)\}_{i=1}^{S}$, we can construct an estimator $\hat{k}_f(m)$ for $k_{f}(m)$ as shown below:
\begin{equation}\label{eq:kfmhat}
\hat{k}_f(m) = \frac{1}{S(S-1)}\sum_{i\neq j}f(x_i,x_j,\mathbf{g}^i,\mathbf{g}^j, m).\\
\end{equation}
When the number of samples $S$ tends to infinity, the estimator $\hat{k}_f(m)$ converges to $k_f(m)$.

In this work, we investigate a special kind of correlation function in the form of
\begin{equation}\label{eq:irrtyp}
f_A(x,y,\mathbf{g}^1,\mathbf{g}^2, m)=\Tr(B_{xy}[\tau_1 (g^{1}_{m})\otimes \tau_2(g^2_m)] \prod_{i=1}^{m-1} A[\tau_1(g^1_i)\otimes \tau_2(g^2_i)]),
\end{equation}
where $\tau_1$ and $\tau_2$ are both irreducible representations of group $\mathbb{G}$ and $A$, $B_{xy}$ are operators supported on the corresponding representation space of $\tau_1\otimes \tau_2$. In this case, the expectation value $k_f(m)$ is denoted as $k_A(m)$ since it encodes information of the operator $A$. For the special-type correlation function of Eq.~\eqref{eq:irrtyp}, Theorem~\ref{th:exdec1} tells us that the expectation value $k_A(m)$ exhibits an exponential decay behavior with respect to the sequence length $m$. The complete proof is available in Appendix~\ref{appendssc:thm1pf}.

\begin{theorem}\label{th:exdec1}
(Exponential decay) Given a generalized UIRS protocol with group $\mathbb{G}$, the expectation value $k_A(m)$ of the second-order correlation function given by Eq.~\eqref{eq:irrtyp} has the following form:
\begin{equation}
    k_A(m)=\Tr(\Theta(\{E_x\},\rho)[\Phi(A,\Lambda)]^{m-1}),
\end{equation}
where $\Theta(\{E_x\},\rho)$ and $\Phi(A,\Lambda)$ are induced matrices dependent on SPAM and channel $\Lambda$, respectively. 
If $\omega(g)=\tau_1(g)^{n_{\tau1}}\oplus \omega'_1(g)$ where $\omega'_1(g)$ contains no copy of $\tau_1(g)$ and $\omega(g)=\tau_2(g)^{n_{\tau2}}\oplus \omega'_2(g)$ where $\omega'_2(g)$ contains no copy of $\tau_2(g)$, then the matrix $\Phi(A,\Lambda)$ has the element
\begin{equation}\label{eq:decayshadow}
[\Phi(A,\Lambda)]_{ii',jj'}=\frac{1}{|P_j||P_{j'}|}\Tr((P_i\otimes P_{i'})A^{T}(P_{j}\otimes P_{j'})\Lambda^{\otimes 2}),
\end{equation}
where $P_i$ and $P_j$ are the projectors onto the $i$-th and $j$-th copies of $\tau_1$ inside the representation $\omega$, respectively, $P_{i'}$ and $P_{j'}$ are the projectors onto the $i'$-th and $j'$-th copies of $\tau_2$ inside the representation $\omega$, respectively. In the special case that the representation $\omega$ only has one copy of $\tau_1$ and one copy of $\tau_2$, the matrix $\Phi(A,\Lambda)$ reduces to a number.
\end{theorem}

As an application of this generalized UIRS protocol, we demonstrate that OTOC can be estimated via the generalized UIRS protocol robust to SPAM errors. We first briefly review the definition of OTOC, which originates from information scrambling. Information scrambling is one of the signatures of quantum chaos in many-body systems, and the strength of the scrambling can be quantitatively characterized by OTOC~\cite{Swingle2018Unscrambling,PhysRevX.9.031048}. Specifically, for a local operator $V$, after a unitary dynamics $U_t=e^{-iHt}$ controlled by a Hamiltonian $H$, the local information encoded by $V$ can spread over many sites and become nonlocal. Assume there is another operator $W$ located at a different position from $V$ and define $V(t)=U_t^{\dagger} V U_t$. Initially, $V = V(0)$ is localized and commutes with $W$, but $V(t)$ will be nonlocal and become non-commutative with $W$ when $V(t)$ spread to the position of $W$. Consequently, the operator growth associated with $U_t$ can be characterized by the expectation value of the squared commutator of $V(t)$ and $W$ over a state $\rho$,
\begin{equation}
    C(t)=\langle [V(t),W]^{\dagger}[V(t),W]\rangle_{\rho}.
\end{equation}
A closely related quantity is OTOC, defined by 
\begin{equation}
  O(t)=\langle W^{\dagger}V(t)^{\dagger}WV(t) \rangle_{\rho},
\end{equation}
which satisfies 
\begin{equation}
   C(t)=2(1-\Re(O(t))).
\end{equation}
The quantities $O(t)$ and $C(t)$ both measure the strength of information scrambling. Here, we take $\rho$ as the maximally mixed state, which can also be regarded as the thermal state of infinite temperature. In this case, OTOC has an explicit form:
\begin{equation}
O(t)=\frac{1}{d}\Tr(W^{\dagger}V(t)^{\dagger}WV(t)).
\end{equation}

Below, we focus on the multi-qubit system with dimension $d=2^n$. We consider $V$ and $W$ to be nontrivial and unnormalized Pauli operators,  $V, W \in \sqrt{d}\mathbb{P}_n^0$. Notice that the situation can also be generalized straightforwardly to the case that $V$ and $W$ are arbitrary operators. In the following theoretical analysis, we assume that the random gate is taken from the multi-qubit Clifford group and is noiseless. To estimate the OTOC robustly against SPAM errors, we introduce an adjustment of the generalized UIRS protocol by inserting a unitary gate $\mathcal{U}_t(\rho)=U_t \rho U_t^{\dagger}$ between two random Clifford gates. This can be interpreted as considering the unitary evolution $U_t$ as the noise, as depicted in Figure~\ref{fig:u2}. In reality, the Clifford gates $\{g_i\}$ are also noisy, and we can regard their noises as being absorbed by the unitary gate. For instance, when considering $\phi(g_i) = \Lambda_L \omega(g_i)\Lambda_R$ like Eq.~\eqref{eq:noisygi}, we can view the inserted gate to be a noisy implementation of the unitary gate, $\widetilde{\mathcal{U}}_t=\Lambda_R \mathcal{U}_t \Lambda_L$, which allows us to obtain the information of the noisy unitary dynamics robust to SPAM errors. In this sense, the implemented Clifford gate is ideal and then our protocol can estimate the OTOC of the noisy unitary evolution without the influence of Clifford gate noises.

\begin{figure}[!htbp]
 \[
 \Qcircuit @C=1em @R=.7em {
  \lstick{} & \gate{g_1} & \gate{U_t} & \gate{g_2} & \gate{U_t} & \qw & \cdots & & \gate{g_i} & \gate{U_t} & \gate{g_{i+1}} & \qw & \cdots & & \gate{U_t} & \gate{g_m} & \qw
 }
 \]
\caption{Random Clifford gate sequence $g_1, g_2, \cdots, g_m$ intertwined with a fixed unitary gate, $U_t$.}
\label{fig:u2}
\end{figure}
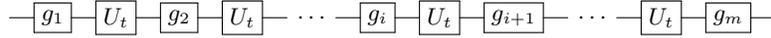

For estimating OTOC, we define the second-order correlation function as 
\begin{equation}\label{eq:otoccor}
\begin{split}
f_A(x,y,\mathbf{g}^1,\mathbf{g}^2, m)=\Tr(B_{xy}\tau_{ad}(g_m^1)\otimes \tau_{ad}(g_m^2)\prod_{i=1}^{m-1} A \tau_{ad}(g_i^1)\otimes \tau_{ad}(g_i^2)),
\end{split}
\end{equation}
where $A=(d^2-1)^2\sum_{\sigma\in \mathbb{P}_n^0}\Tr(W \sigma W\sigma)\lketbra{\sigma\otimes \sigma}{V \otimes V}$, $B_{xy}= \lketbra{\rho}{E_{x}}\otimes \lketbra{\rho}{E_{y}}$, and $\tau_{ad}$ is the adjoint representation supported on the traceless matrices. After obtaining shadow sequence data, we can estimate the expectation value $k_A(m)$. Next, we show the relationship between $k_A(m)$ and OTOC.
\begin{corollary}\label{th:exdec3}
The expectation value $k_A(m)$ of the second-order correlation function defined in Eq.~\eqref{eq:otoccor} has an exponential decay behaviour with respect to the sequence length $m$, $k_A(m)=a (dO(t))^{m-1}$. The decay parameter is $p(A)=dO(t)$ with $O(t)=\frac{1}{d}\Tr(W^{\dagger}V(t)^{\dagger}WV(t))$ being OTOC.
\end{corollary}

The full proof of Corollary~\ref{th:exdec3} is available in Appendix~\ref{appendssc:thm2proof}. Thus, OTOC can be evaluated robustly with the protocol proposed in this section, with the simulation shown in Section~\ref{sc:simulation}. 

To analyze the sample complexity of this protocol, we evaluate the variance of the correlation function. 
Viewing $f_A$ defined in Eq.~\eqref{eq:otoccor} as a random variable with probability distribution from the shadow data, $p(x, \mathbf{g}^1)p(y, \mathbf{g}^2)$, the variance of this random variable is upper bounded by $O(d^{8m-12})$.
The full proof of the variance is detailed in Appendix~\ref{appendssc:var}. In the experiment, we can estimate the expectation value $k_A(m)$ through statistics method such as Eq.~\eqref{eq:kfmhat} when obtaining $S$ samples of shadow data. Particularly, since $k_A(2)=adO(t)$ and $k_A(1)=a$, $k_A(2)/k_A(1)$ is equivalent to the value of OTOC if ignoring a constant $d$. Hence, we can estimate OTOC by the ratio $\hat{k}_A(2,S)/\hat{k}_A(1,S)$ with $S$ samples. The variance of the estimators $\hat{k}_A(2,S)$ and $\hat{k}_A(1,S)$ can be upper bounded by $O(d^{4}S^{-2})+ O(S^{-1})$ and $O(d^{-4}S^{-2})+O(d^{-1}S^{-1})$, respectively. Then the variance of the ratio has an upper bound given by
\begin{equation}
\mathbf{Var}\left(\frac{\hat{k}_A(2,S)}{\hat{k}_A(1,S)}\right) \leq \left(\frac{\mathbb{E}(\hat{k}_A(2,S))}{\mathbb{E}(\hat{k}_A(1,S))}\right)^2\left[\frac{\mathbf{Var}( \hat{k}_A(2,S))}{(\mathbb{E}(\hat{k}_A(2,S)))^2}+\frac{\mathbf{Var} (\hat{k}_A(1,S))}{(\mathbb{E}(\hat{k}_A(1,S)))^2}\right] =  O\left(\frac{d^{8}}{S^{2}}\right)+O\left(\frac{d^{5}}{S}\right),
\end{equation}
where the two expectation values can be evaluated by $\mathbb{E}(\hat{k}_A(2,S)) = O(d^{-1})$ and $\mathbb{E}(\hat{k}_A(1,S)) = O(d^{-2})$. The detailed calculation of the variances and expectations can be found in Appendix~\ref{appendssc:var}. Thus, the number of samples $S$ only needs to take $O(\frac{d^5}{\epsilon^2})$ to evaluate the estimate within precision $\epsilon$.

\subsection{Nonlinear correlation via identical sequences}\label{sc:ncrandom}
The nonlinear correlation function defined in Eq.~\eqref{eq:independcorr} is connected with two independent pieces of sequences. In this section, we set the two gate sequences to be identical and study the corresponding second-order correlation function, where we find that the protocol estimating unitarity robustly presented in Ref.~\cite{Wallman2015coherence} is a specific instance of this protocol.

Specifically, we consider the second-order correlation function when the two random gate sequences are chosen to be the same, i.e., $\mathbf{g}^1=\mathbf{g}^2$, defining
\begin{equation}\label{eq:expec}
k_f(m)=\sum_{x,y\in\mathcal{X}}\mathbb{E}_{\mathbf{g}\in \mathbb{G}^{\times m}} f(x,y,\mathbf{g},m)p(x,\mathbf{g})p(y,\mathbf{g}).
\end{equation}
Similarly to the case for independent sequences, we assume
\begin{equation}\label{eq:groupfunction}
f(x,y,\mathbf{g},m)=f_A(x,y,\mathbf{g}, m)=\Tr(B_{xy}\tau(g_m)\prod_{i=1}^{m-1}A\tau(g_i)),
\end{equation}
where $\tau$ is an irreducible representation of the group $\mathbb{G}$ and $A$, $B_{xy}$ are matrices supported on the representation space associated with $\tau$. The following theorem states that the expectation value $k_A(m)=k_{f_A}(m)$ has an exponential decay, with proof shown in Appendix~\ref{appenssc:thm3proof}.

\begin{theorem}\label{th:exdec2}
(Exponential decay) Given a generalized UIRS protocol with respect to the group $\mathbb{G}$, the expectation value $k_A(m)$ defined with Eqs.~\eqref{eq:expec} and~\eqref{eq:groupfunction} has the following form:
\begin{equation}
    k_{A}(m)=\Tr(\Theta(\{E_x\},\rho)[\Phi(A,\Lambda)]^{m-1}),
\end{equation}
where $\Theta(\{E_x\},\rho)$ and $\Phi(A,\Lambda)$ are induced matrices dependent on SPAM and channel $\Lambda$, respectively. The matrix $\Phi(A,\Lambda)$ has the element
\begin{equation}\label{eq:decayrandom}
    [\Phi(A,\Lambda)]_{i,j}=\frac{1}{|P_j|}\Tr(P_i A^{T}P_{j} \Lambda^{\otimes 2}),
\end{equation}
where $P_i$ and $P_j$ are the projectors onto the $i$-th and $j$-th copies of $\tau$ inside $\omega^{\otimes 2}$, respectively.
In the special case that the representation $\omega^{\otimes 2}$ only has one copy of $\tau$, the matrix $\Phi(A,\Lambda)$ reduces into a number.
\end{theorem}

Below, we show that the protocol in this section can be used to estimate unitarity. For the between-gates noise channel $\Lambda:\mathcal{B}(\mathbb{C}^d)\rightarrow \mathcal{B}(\mathbb{C}^d)$, the unitarity of the channel is defined as the average purity of the output states with the identity part being subtracted,
\begin{equation}
   u (\Lambda)=\frac{d}{d-1} \int d\psi \Tr[\Lambda'(\ketbra{\psi})^{\dagger}\Lambda'(\ketbra{\psi})],
\end{equation}
where $\Lambda'(A)=\Lambda(A)-\frac{\Tr\Lambda(A)}{d}\mathbf{1}$.

The unitarity can be estimated efficiently and robustly against SPAM errors through an experimental protocol based on randomized benchmarking \cite{Wallman2015coherence}. Here, we prove that unitarity can be estimated under the generalized UIRS protocol when the two involved random sequences are identical and the group $\mathbb{G}$ is the Clifford group. Again, we assume the noise channel of the Clifford gates to be gate-independent.

Given a state, $\rho$, a sequence of gates $\{g_1, g_2, \cdots, g_m\}$, and an observable, $E = \sum_{x} E_{x}$ with $\{E_x\}_{x}$ representing a computational-basis measurement, it has been proven that the square of the measurement result has an exponential decay after taking an expectation over the Clifford group $\mathbb{G}$ \cite{Wallman2015coherence},
\begin{equation}\label{eq:unitarity}
\underset{\mathbf{g}\in \mathbb{G}^{\times m}}{\mathbb{E}}\Big[\Tr (\tilde{E} \prod_{i=1}^{m}\phi(g_i)[\tilde{\rho}])\Big]^2=a+b u(\Lambda)^{m-1},
\end{equation}
where $a$ and $b$ are fitting constants and $u(\Lambda)$ is the unitarity of the noise channel. The notation $\widetilde{\cdot}$ represents the noisy versions of quantum states, gates, and observables.

We demonstrate that Eq.~\eqref{eq:unitarity} can be written in the form of Eq.~\eqref{eq:expec} by properly selecting the operator $A$ and $B_{xy}$. Let $A=(\sum_{\sigma\in \mathbb{P}_n} \lketbra{\sigma}{\sigma})^{\otimes 2}$ as the identity in the double representation $\omega^{\otimes 2}$ and $B_{xy}= \lket{\mathbf{1}}\lbra{E_x}\otimes \lket{\mathbf{1}}\lbra{E_y}$, and set the irreducible representation as the trivial 
representation of the Clifford group. Then, the second-order correlation function is 
\begin{equation}
f_A(x,y,\mathbf{g}, m)=\Tr( \lket{\mathbf{1}}\lbra{E_x}\otimes \lket{\mathbf{1}}\lbra{E_y}\prod_{i=1}^{m}\tau_{tr}(g_i))=\Tr(E_x)\Tr(E_y)=1.
\end{equation}
Therefore,
\begin{equation}
\begin{split}
k_A(m) &= \underset{\mathbf{g}\in \mathbb{G}^{\times m}}{\mathbb{E}}\sum_{x,y} f_A(x,y,\mathbf{g},m)p(x,\mathbf{g})p(y,\mathbf{g})\\
&= \underset{\mathbf{g}\in \mathbb{G}^{\times m}}{\mathbb{E}}\sum_{x,y} \Tr(\tilde{E}(x)\prod_{i=1}^{m}\phi(g_i)[\tilde{\rho}])\Tr(\tilde{E}(y)\prod_{i=1}^{m}\phi(g_i)[\tilde{\rho}])\\
&=\underset{\mathbf{g}\in \mathbb{G}^{\times m}}{\mathbb{E}}\left[\Tr(\tilde{E} \prod_{i=1}^{m}\phi(g_i)[\tilde{\rho}]) \right]^2.
\end{split}
\end{equation}

According to Theorem~\ref{th:exdec2}, the quantity $k_A(m)$ has an exponential decay with respect to the matrix $\Phi(A,\Lambda)=P_{\tau}(A\otimes \Lambda^{\otimes 2})P_{\tau}$ where $P_{\tau}$ is the projector onto the trivial subspace within $\omega^{\otimes 2}$. The explicit form of the projector is
\begin{equation}\label{eq:ptau}
P_{\tau}=\lketbra{B_1\otimes B_1}{B_1\otimes B_1}+\lketbra{B_1\otimes B_2}{B_1\otimes B_2},
\end{equation}
where the two operators are
\begin{equation}\label{eq:B2}
\begin{split}
     B_1&=\frac{\mathbf{1}_{d^2}}{d}, B_2=\frac{F-B_1}{\sqrt{d^2-1}}.
\end{split}
\end{equation}
Here, $F$ is a SWAP operator on the space of $\Lambda^{\otimes 2}$. When $\Lambda$ is trace-preserving, 
\begin{equation}
    \Phi(A,\Lambda)=\begin{pmatrix} \lbra{B_1}\Lambda^{\otimes 2}\lket{B_1} & \lbra{B_1}\Lambda^{\otimes 2}\lket{B_2} \\ \lbra{B_2}\Lambda^{\otimes 2}\lket{B_1} & \lbra{B_2}\Lambda^{\otimes 2}\lket{B_2} \end{pmatrix}=\begin{pmatrix} 
    1 & \lbra{B_1}\Lambda^{\otimes 2}\lket{B_2} \\ 0 & u(\lambda) \end{pmatrix}
\end{equation}
is a $2\times 2$ matrix. And the two eigenvalues of matrix $\Phi$ are $1$ and $u(\lambda)$. Thus, there exist constants $a$ and $b$ such that
\begin{equation}\label{eq:unifit}
k_A(m)=a+b u(\Lambda)^{m-1}.
\end{equation}

For different $m$, estimate the expectation values $k_A(m)$ and fit these values to Eq.~\eqref{eq:unifit}. Thus, we can estimate the unitarity robustly and efficiently only with a single exponential fitting.

\subsection{Measurable functions via UIRS}\label{ssc:measure}
In Sections~\ref{sc:ncshadow} and~\ref{sc:ncrandom}, we have derived the exponential forms of the expectation values as shown in Theorems~\ref{th:exdec1} and~\ref{th:exdec2}. We refer to the two different protocols as independent UIRS and identical UIRS. It is worth mentioning that in our protocol, only functions in the form of Eq.~\eqref{eq:decayshadow} can be directly measured with independent UIRS, and only functions in the form of Eq.~\eqref{eq:decayrandom} can be directly measured with identical UIRS. Due to the existence of the projectors in Eqs.~\eqref{eq:decayshadow} and~\eqref{eq:decayrandom}, the nonlinear function $\Tr(A^T\Lambda^{\otimes 2})$ can only be measured when $A$ is inside the span of projectors. Also, owing to the different classical postprocessing procedures of the independent UIRS protocol and the identical UIRS protocol, the projectors in Eqs.~\eqref{eq:decayshadow} and~\eqref{eq:decayrandom} differ from each other. There might exist some operators that can only be measured by the independent UIRS and some other operators that can only be measured by the identical UIRS, which depends on the choice of the gate set $\mathbb{G}$. Nonetheless, if we assume the matrix exponential fitting is feasible and the random gates are taken from the $n$-qubit Clifford group, we can show that any operators that identical UIRS can measure can also be measured via independent UIRS. Below, we delve into this discussion in detail.

From Eq.~\eqref{eq:decayshadow}, when fixing two irreducible representation $\tau$ and $\tau'$ and a matrix $A$, one can obtain the value of $\Tr((P_i\otimes P_{i'})A^{T}(P_{j}\otimes P_{j'})\Lambda^{\otimes 2})$, where $P_i$ and $P_j$ are projectors onto the irreducible representation of $\tau$ inside the Liouville representation $\omega$ and $P_{i'}$ and $P_{j'}$ are projectors onto the irreducible representation of $\tau'$ inside $\omega$. Thus, independent UIRS can measure observables in the span of $\{P_{\tau} \otimes P_{\tau'}\}$ where $P_{\tau}$ and $P_{\tau'}$ are projectors onto the irreducible representation space associated with $\tau$ and $\tau'$ in $\omega$, respectively. $\tau$ and $\tau'$ are arbitrary irreducible representations of group $\mathbb{G}$. If $\omega$ contains several copies of $\tau$, then $P_{\tau} = \sum_i P_i$ is the sum of the projectors onto each copy of $\tau$ where $P_i$ is the projector onto the $i$-th copy of $\tau$ inside the representation $\omega$. We denote the span of the projectors in this case as $\mathcal{S}^{\mathbb{G}}_1 = \text{span}\{P_{\tau} \otimes P_{\tau'}\}$.

Similarly, from Eq.~\eqref{eq:decayrandom}, identical UIRS can only measure observables in the span of $\{P_{\nu}\}$ where $P_{\nu}$ is a projector onto the irreducible representation space associated with $\nu$ in $\omega^{\otimes 2}$. $\nu$ is an irreducible representations of group $\mathbb{G}$. We denote the span of the projectors in this case as $\mathcal{S}^{\mathbb{G}}_2 = \text{span}\{P_{\nu}\}$.

Note that $\mathcal{S}^{\mathbb{G}}_1$ and $\mathcal{S}^{\mathbb{G}}_2$ are related to the group $\mathbb{G}$ and would be larger if $\mathbb{G}$ is larger. Below, we consider the case that $\mathbb{G}$ is the $n$-qubit Clifford group. In this case,
\begin{align}
\label{eq:projshadow}\mathcal{S}^{\mathbb{G}}_1 &= \text{span}\{P_{tr} \otimes P_{tr}, P_{tr} \otimes P_{ad}, P_{ad} \otimes P_{tr}, P_{ad} \otimes P_{ad}\};\\
\label{eq:projrandom}\mathcal{S}^{\mathbb{G}}_2 &= \text{span}\{P_{d}, P_{id}, P_{r}, P_{l}, P_{[S]}, P_{\{S\}}, P_{[A]}, P_{\{A\}}\}.
\end{align}
Here, $P_{tr}$ and $P_{ad}$ are two distinct irreducible representations of the Clifford group in $\omega$ as introduced in Section~\ref{ssc:Clifford}. $P_{d}, P_{id}, P_{r}, P_{l}, P_{[S]}, P_{\{S\}}, P_{[A]}$, and $P_{\{A\}}$ are eight different irreducible representations of the Clifford group in $\omega^{\otimes 2}$~\cite{Helsen2018Representations}. Then, the operators that can be measured via the independent UIRS and the identical UIRS are determined by Eqs.~\eqref{eq:projshadow} and~\eqref{eq:projrandom}, respectively.

For OTOC, the first example in this work, the observable $A$ is $(d^2-1)^2\sum_{\sigma\in \mathbb{P}_n^0} \Tr(W \sigma W\sigma) \lketbra{\sigma\otimes \sigma}{V^{\dagger}\otimes V}$ satisfying $A = P_{ad}^{\otimes 2}AP_{ad}^{\otimes 2}$. Meanwhile, $A = P_dAP_d$. Thus, OTOC can be measured by both the independent UIRS and the identical UIRS. But the difference is that the former only requires the single-exponential fitting, and the latter requires the matrix-exponential fitting as $P_d$ contains multiple copies of an irreducible representation of $\mathbb{G}$. For unitarity, as discussed in the previous section, this quantity can be measured via the identical UIRS. Meanwhile, the unitarity is equal to $\lbra{B_2}\Lambda^{\otimes 2}\lket{B_2}$ with $B_2$ defined in Eq.~\eqref{eq:B2} and corresponds to the observable of $\lketbra{B_2}{B_2}$. Note that $\lketbra{B_2}{B_2}$ is inside $P_{ad} \otimes P_{ad}$ so unitarity can be also measured via the independent UIRS. The difference between the two protocols is the classical postprocessing and the sample complexity.

Though the two examples in this work can be measured via both two protocols, there exist examples that can only be measured via independent UIRS. For instance, the observable $P_{ad}^{\otimes 2}$ can only be measured via the independent UIRS, as it does not belong to $\mathcal{S}^{\mathbb{G}}_2$. Conversely, we found that there is no example that can only be measured via identical UIRS for $n$-qubit Clifford group as $\mathcal{S}_2^{\mathbb{G}} \subseteq \mathcal{S}_1^{\mathbb{G}}$. This can also be easily seen from the following equations.
\begin{equation}
    \begin{split}
P_{id} &= P_{tr}\otimes P_{tr}\\
P_l &= P_{ad}\otimes P_{tr}\\
P_r &= P_{tr}\otimes P_{ad}\\
P_{d} &= P_{d} \cdot (P_{ad}\otimes P_{ad})\\
P_{[S]} &= P_{[S]} \cdot (P_{ad}\otimes P_{ad})\\
P_{\{S\}} &= P_{\{S\}} \cdot (P_{ad}\otimes P_{ad})\\
P_{[A]} &= P_{[A]} \cdot (P_{ad}\otimes P_{ad})\\
P_{\{A\}} &= P_{\{A\}} \cdot (P_{ad}\otimes P_{ad}).
    \end{split}
\end{equation}
Thus, for the $n$-qubit Clifford group, $\mathcal{S}_2^{\mathbb{G}}\subset \mathcal{S}_1^{\mathbb{G}}$ and we conclude that the observables that independent UIRS can measure contain the observables that identical UIRS can measure. This property comes from the fact the Liouville representation only contains a trivial representation and another irreducible representation for the Clifford group. Due to the same reason, the property that independent UIRS can measure more observables than identical UIRS holds for any unitary 2-design group~\cite{unitarydesignDankert2009}. Nonetheless, this phenomenon does not hold for any group. For the Pauli group, identical UIRS can measure more than independent UIRS, as $\mathcal{S}_1^{\mathbb{G}}\subset \mathcal{S}_2^{\mathbb{G}}$ in this case.

\section{Simulation results and analysis}\label{sc:simulation}
In this part, we present the simulation results leveraging the generalized UIRS protocol, primarily focusing on the estimation of OTOC. OTOC is defined with a Hamiltonian evolution and two observables $V$ and $W$. In our simulation, the evolution is a unitary dynamics, $U_t = e^{-iHt}$, where $H$ is a disordered Ising interaction~\cite{Ising2004},
\begin{equation}
H_{\text{disordered Ising}}=\sum\limits_{i<j}J_{i,j}\sigma_i^x\sigma_j^x+\frac B2\sum\limits_i\sigma_i^z+\sum\limits_i\frac{D_i}2\sigma_i^z,
\end{equation}
where $J_{i,j}=\frac{J_0}{|i-j|^\alpha}$, $D_i$ is uniformly and randomly chosen from $[-D_{\mathrm{max}},D_{\mathrm{max}}]$. And we set the random Clifford gates to be noiseless in the simulation. The observables $V$ and $W$ are set as Pauli operators with $\sigma_y$ on the last qubit and $\sigma_x$ on the last but one qubit, respectively. Moreover, the POVM $\{E_x\}$ is chosen as the computational-basis measurement and the initial state $\rho$ is assigned as $\ketbra{0}$.

As introduced in Eq.~\eqref{eq:otoccor}, to evaluate OTOC, we first calculate the correlation function $f_A(x,y,\mathbf{g}^1,\mathbf{g}^2, m)$ by sampling two sequences of Clifford gates with length $m$, denoted as $(\mathbf{g}^1,\mathbf{g}^2)$. To reduce the computational difficulties, we simplify the expression of Eq.~\eqref{eq:otoccor} with results shown in Appendix~\ref{appendixssc:formula}. Then we simulate Eq.~\eqref{eq:expcor} by summing over all measurement results $x$ and $y$ to obtain an estimator $\hat k_A(m,\mathbf{g}^1,\mathbf{g}^2)=\sum_{x,y}f_A(x,y,\mathbf{g}^1,\mathbf{g}^2,m)p(x,\mathbf{g}^1)p(y,\mathbf{g}^2)$. Then, we repeat sampling $(\mathbf{g}^1,\mathbf{g}^2)$ for $S$ times to obtain the expectation, which we denote as $\hat k_A(m,S)$.

As we propose, the value of OTOC is related to the decay of $k_A(m)$. Here, we take $m$ to be $1$ and $2$ and employ the ratio of $\frac{k_A(2)}{k_A(1)}$ to derive the OTOC estimate. From the simulation process above, we obtain one OTOC estimate $x_1(S)=\frac{\hat{k}_A(2,S)}{\hat{k}_A(1,S)}$. Subsequently, through the repetition of this process, we accumulate a series of estimates $\{x_1(S),\cdots,x_N(S)\}$. The ultimate OTOC estimate is then defined as $\bar{x}(S)=\frac{1}{N}\sum_{i=1}^Nx_i(S)$. We illustrate the estimated OTOC for 3, 4, and 5 qubit systems in Fig.~\ref{fig:3q}, Fig.~\ref{fig:4q}, and Fig.~\ref{fig:5q}. The results show that our protocol can estimate OTOC accurately. Moreover, if we denote the variance of $\bar x(S)$ as $s^2$ and the variance of $x_i(S)$ as $\sigma^2$, then we have $s^2=\frac{\sigma^2}{N}$. We show the variance of the estimate $\bar{x}(S)$ with respect to the sampling sequence number $S$ for different qubit systems in Fig.~\ref{fig:var}.

\begin{figure}[htbp]
    \centering
    \subfigure[OTOC for 3 qubits]{\includegraphics[width=8cm]{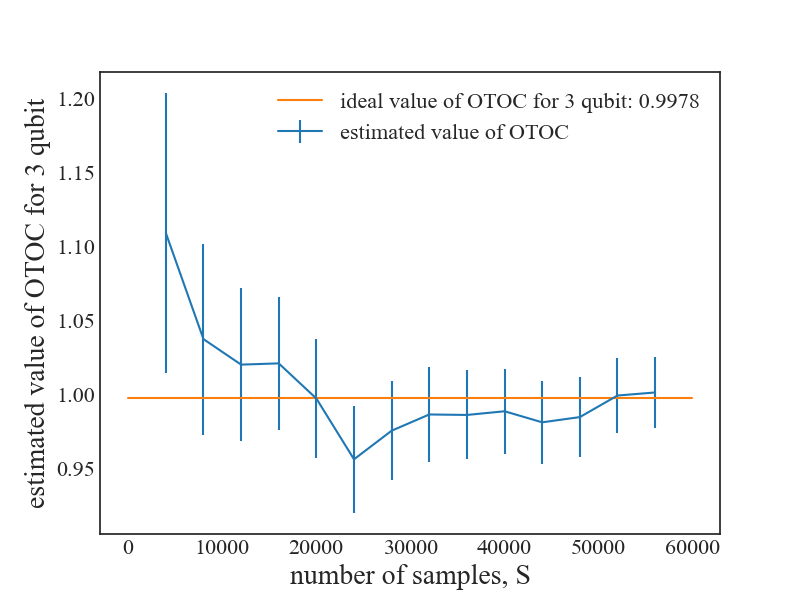}\label{fig:3q}}
    \subfigure[OTOC for 4 qubits]{\includegraphics[width=8cm]{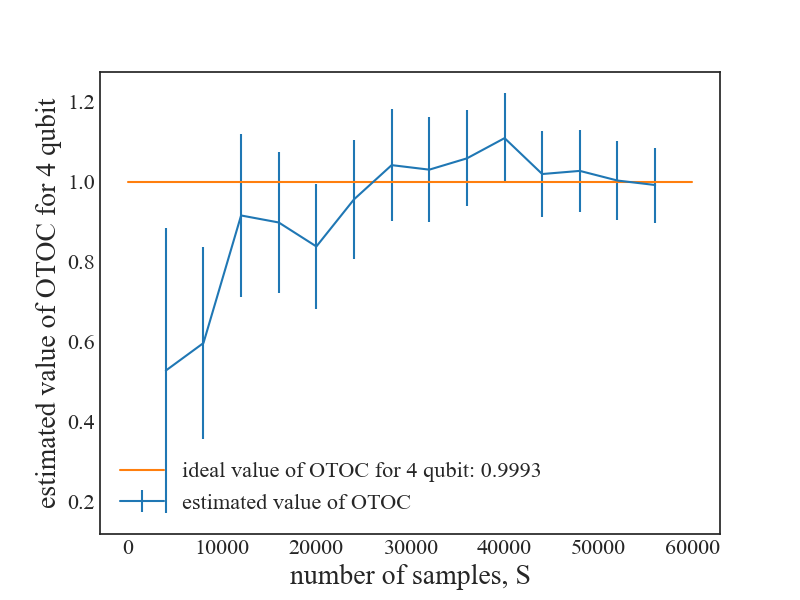}\label{fig:4q}}
    \subfigure[OTOC for 5 qubits]{\includegraphics[width=8cm]{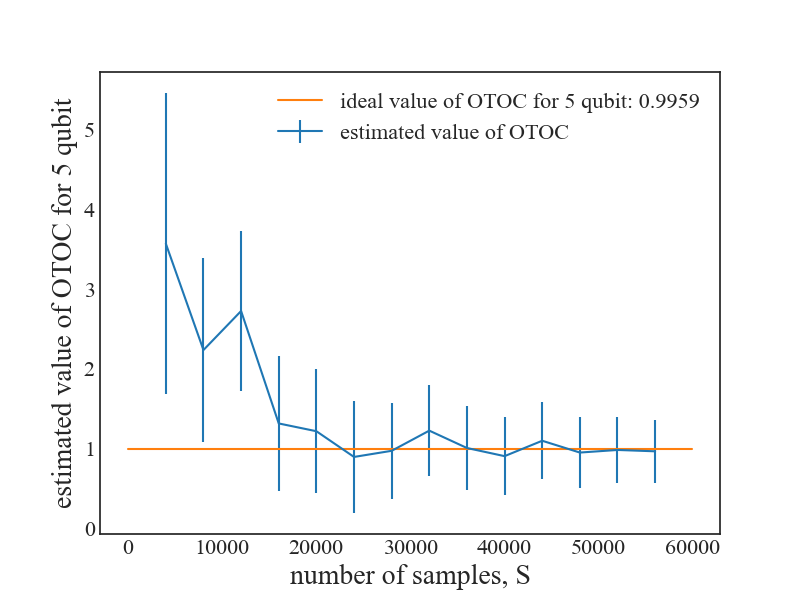}\label{fig:5q}}   
    \subfigure[Variance for 2, 3, 4, and 5 qubit]{\includegraphics[width=8cm]{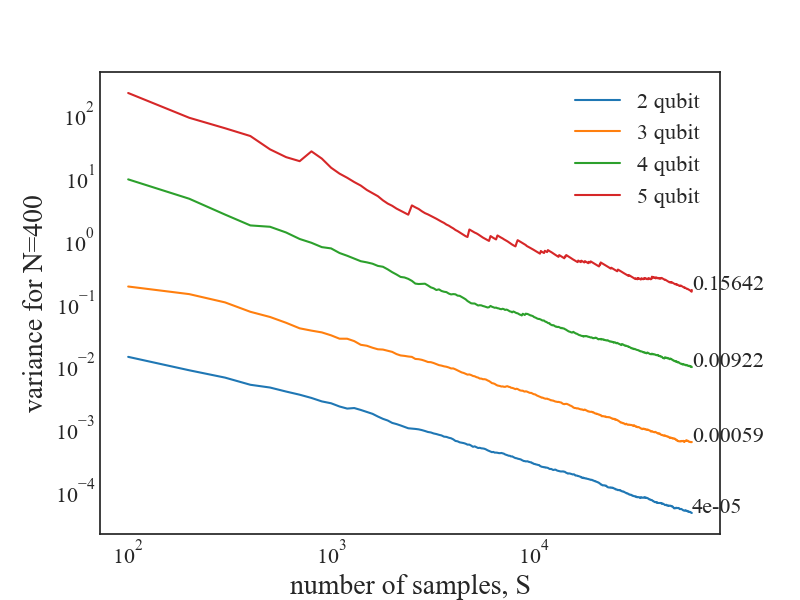}\label{fig:var}}
    \caption{Convergence of estimated OTOC for (a) 3 qubits, (b) 4 qubits, and (c) 5 qubits.  In the figures, the horizontal axis corresponds to $S$, the number of sampling sequences in $\hat k_A(m,S)$. The orange horizontal line denotes the theoretical OTOC value in the ideal case that we want to obtain. And the blue line represents the average $\bar x$ over $\{x_1,\cdots,x_N\}$, for $N=400$ and $S_{max}=60000$ here.  Moreover, the blue vertical line represents the standard deviation for $\bar x(S)$. Due to the substantial size of the Clifford group, especially for qubit numbers beyond 2, exhaustive sampling from the entire Clifford group proves challenging. Instead, we opt to randomly select a subgroup from the $n$-qubit Clifford group and sample from this subgroup for simulation purposes. Moreover, owing to the extensive size of the Clifford group, the estimate $x_i(S)$ exhibits considerable variation, and $\bar x(S)$ deviates from the ideal OTOC when the number of samples $S$ is relatively small. As the sample size $S$ increases, the estimated result gradually converges towards the ideal value.
    (d) In this figure, we illustrate the variance of $\bar x(S)$. As $\bar x(S)$ is the mean value of $\{x_1(S),\cdots,x_N(S)\}$, the relationship between their variance is that $\mathrm{Var}(\bar x(S))=\frac{\mathrm{Var}(x_i(S))}{N}$, where $N=400$ is the size of the estimates. In this figure, the horizontal axis is the logarithm of the sample number $S$, and the vertical axis represents the logarithm of the variance for different qubits.}
\end{figure}


Furthermore, since OTOC depends on the evolution time, we demonstrate how the estimated OTOC changes with time. We sample ten timestamps with equal intervals and obtain the estimated OTOC for each time, as shown in Fig.~\ref{fig:time}. The results underscore that the estimate captures the temporal evolution of the OTOC, thereby presenting a valuable tool for investigating quantum scrambling phenomena within quantum many-body systems.

As we proposed, our method is robust to SPAM error, and we compare it with another method using statistical correlation in~\cite{Vermersch_2019}, as illustrated in Fig.~\ref{fig:compare}. The statistical correlation method consists of applying a global unitary to an arbitrary state and then separately measuring $\langle W(t)\rangle$ and $\langle V^{\dagger}W(t)V\rangle$ to obtain the OTOC.
For the SPAM error, we introduce the depolarizing channel after state preparation and before measurement. For density matrix $\rho$ and error probability $p$, the state undergoes a transformation to $(1-p)\rho+\frac{p}{d} \mathbf{1}$ after the depolarizing channel. The simulation result highlights the advantageous performance of our protocol, particularly when 
 SPAM errors are large.

\begin{figure}
    \centering
    \subfigure[OTOC with time]{\label{fig:time}\includegraphics[width=8cm]{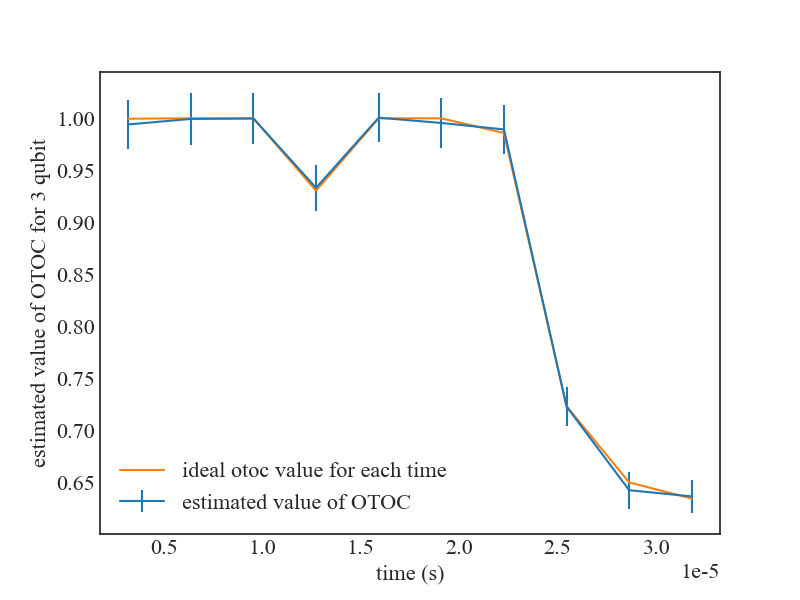}}
    \subfigure[Compared OTOC value under SPAM error]{\label{fig:compare}\includegraphics[width=8cm]{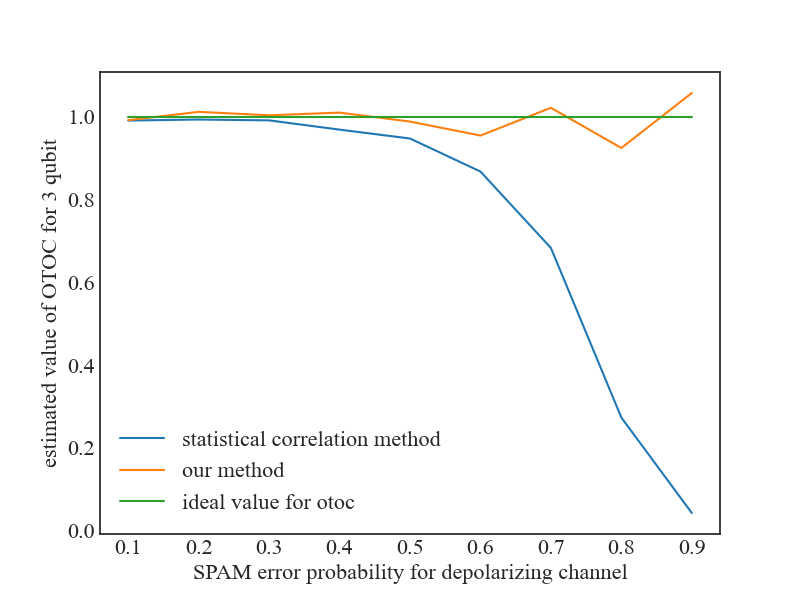}}
    \caption{(a) In this figure, the horizontal axis is the time, and the orange line is the ideal OTOC value. The blue line is the mean value $\bar x(S)$ for $N=400, S=30000$, and the standard deviation of $\bar x(S)$ is shown as the error bar. (b) In this figure, the horizontal axis is the probability $p$ in the depolarizing channel, after which the state $\rho$ undergoes the transformation $(1-p)\rho+\frac{p}{d}\mathbf{1}$; the green line is the ideal OTOC value; the orange line is the estimated value of OTOC from our method; the blue line is the estimated OTOC from the statistical correlation method~\cite{Vermersch_2019}.} 
\end{figure}

\section{Conclusion}\label{sc:conclusion}
In this work, we propose two generalized UIRS protocols that can be used to estimate non-linear channel properties robust to SPAM errors. In these protocols, we collect shadow data by applying random gate sequences and POVM measurement and calculate the expectation of the correlation function with shadow data from independent or identical gate sequences. Then, the expectation exhibits an exponential decay against the circuit depth, from which we can extract the nonlinear property of quantum channels. 
We show the application of our protocols in evaluating OTOC, an important quantity in quantum many-body systems. From the simulation results, we demonstrate the efficience of our protocol and robustness against SPAM errors. 
As OTOC is directly related to quantum magic~\cite{Nonstabilizerness2023Leone}, our protocol can be used to measure this essential resource in universal quantum computing. Besides OTOC, we also demonstrate the application of our protocol in estimating other properties like unitarity, which characterizes the incoherence of quantum channels.

Note that in this work, we only study second-order correlation functions. The whole protocol can be generalized to the case using three or more pieces of shadow data and be used to investigate higher-order channel properties.
Moreover, the examples in our work mainly focus on the $n$-qubit Clifford group. One can investigate the properties derived with groups $\mathbb{G}$ other than the $n$-qubit Clifford group and explore the possibility that $\mathbb{G}$ is not a group.

\begin{acknowledgements}
The authors would like to thank Wenjun Yu and Pei Zeng for the helpful discussions. This work was supported by the National Natural Science Foundation of China Grant No.~12174216 and the Innovation Program for Quantum Science and Technology Grant No.~2021ZD0300804.
\end{acknowledgements}

\clearpage

\appendix
\section{Proof of Theorem \ref{th:exdec1}}\label{appendssc:thm1pf}
{\bf Theorem 1} {\it (Exponential decay) Given a generalized UIRS protocol with group $\mathbb{G}$, the expectation value $k_A(m)$ of the second-order correlation function given by Eq.~\eqref{eq:irrtyp} has the following form:
\begin{equation}
    k_A(m)=\Tr(\Theta(\{E_x\},\rho)[\Phi(A,\Lambda)]^{m-1}),
\end{equation}
where $\Theta(\{E_x\},\rho)$ and $\Phi(A,\Lambda)$ are induced matrices dependent on SPAM and channel $\Lambda$, respectively. 
If $\omega(g)=\tau_1(g)^{n_{\tau1}}\oplus \omega'_1(g)$ where $\omega'_1(g)$ contains no copy of $\tau_1(g)$ and $\omega(g)=\tau_2(g)^{n_{\tau2}}\oplus \omega'_2(g)$ where $\omega'_2(g)$ contains no copy of $\tau_2(g)$, then the matrix $\Phi(A,\Lambda)$ has the element
\begin{equation}\label{eq:decayshadow}
[\Phi(A,\Lambda)]_{ii',jj'}=\frac{1}{|P_j||P_{j'}|}\Tr((P_i\otimes P_{i'})A^{T}(P_{j}\otimes P_{j'})\Lambda^{\otimes 2}),
\end{equation}
where $P_i$ and $P_j$ are the projectors onto the $i$-th and $j$-th copies of $\tau_1$ inside the representation $\omega$, respectively, $P_{i'}$ and $P_{j'}$ are the projectors onto the $i'$-th and $j'$-th copies of $\tau_2$ inside the representation $\omega$, respectively. In the special case that the representation $\omega$ only has one copy of $\tau_1$ and one copy of $\tau_2$, the matrix $\Phi(A,\Lambda)$ reduces to a number.}
\begin{proof}
\begin{equation}
\begin{split}
k_A(m)&=\sum_{x,y\in\mathcal{X}}\underset{\mathbf{g}^{1},\mathbf{g}^{2}\in \mathbb{G}^{\times m}}{\mathbb{E}}\Tr(B_{xy}[\tau_1(g^{1}_{m})\otimes \tau_2(g^2_m)] \prod_{i=1}^{m-1} A[\tau_1(g^1_i)\otimes \tau_2(g^2_i)]) \lbra{\tilde{E_x}}\prod_{i=1}^{m}\Lambda_L\omega(g^1_i)\Lambda_R\lket{\tilde{\rho}}
\lbra{\tilde{E_y}}\prod_{i=1}^{m}\Lambda_L\omega(g^2_i)\Lambda_R\lket{\tilde{\rho}}\\
&=\sum_{x,y\in\mathcal{X}} \underset{\mathbf{g}^{1},\mathbf{g}^{2}\in \mathbb{G}^{\times m}}{\mathbb{E}}\Tr  \Big((B_{xy}\otimes \lketbra{\Lambda_R(\tilde{\rho})}{\Lambda_L^*(\tilde{E_x})}\otimes \lketbra{\Lambda_R(\tilde{\rho})}{\Lambda_L^*(\tilde{E_y})})\tau_1(g^1_m)\otimes \tau_2(g^2_m)\otimes \omega(g^1_m)\otimes \omega(g^2_m)\times\\
&\prod_{i=1}^{m-1}(A\otimes\Lambda^{\otimes 2})\tau_1(g^1_i)\otimes \tau_2(g^2_i)\otimes \omega(g^1_i)\otimes \omega(g^2_i)\Big)\\
&=\sum_{x,y\in\mathcal{X}}\Tr((B_{xy} \otimes \lketbra{\Lambda_R(\tilde{\rho})}{\Lambda_L^{*}(\tilde{E_x})}\otimes \lketbra{\Lambda_R(\tilde{\rho})}{\Lambda_L^*(\tilde{E_y})})\Big(P_{\tau_1\otimes\tau_2}(A\otimes \Lambda^{\otimes 2})P_{\tau_1\otimes\tau_2}\Big)^{m-1}),
\end{split}
\end{equation}
where $P_{\tau_1\otimes\tau_2}=P_{\tau_1}\otimes P_{\tau_2}=\underset{g_{1}\in \mathbb{G}}{\mathbb{E}}(\tau_1(g_1)\otimes \omega(g_1)) \otimes\underset{g_{2}\in \mathbb{G}}{\mathbb{E}} (\tau_2(g_2)\otimes \omega(g_2))$ is the corresponding representation average projector. If $\omega(g)=\tau_1(g)^{n_{\tau 1}}\oplus \omega'_1(g)$ and $\omega(g)=\tau_2(g)^{n_{\tau 2}}\oplus \omega'_2(g)$, the rank of $P_{\tau_1\otimes\tau_2}$ is $n_{\tau 1}n_{\tau 2}$ and $P_{\tau_1\otimes\tau_2}(A\otimes \Lambda^{\otimes 2})P_{\tau_1\otimes\tau_2}$ can be viewed as a matrix $\Phi$ with dimension $n_{\tau 1}n_{\tau 2}\times n_{\tau 1}n_{\tau 2}$. The element of $\Phi$ is given by
\begin{equation}
\Phi_{i,i',j,j'}=\frac{1}{|P_j||P_{j'}|}\Tr((P_i\otimes P_{i'})A^{T}(P_{j}\otimes P_{j'})\Lambda^{\otimes 2})
\end{equation}
where $P_i$ and $P_j$ are the projectors onto the $i$-th and $j$-th copies of $\tau_1$ inside the representation $\omega$. $P_{i'}$ and $P_{j'}$ are the projectors onto the $i'$-th and $j'$-th copies of $\tau_2$ inside the representation $\omega$. In particular, if $n_{\tau1}=n_{\tau2}=1$, then the matrix $\Phi(A,\Lambda)$ degenerates into a real number, $\Tr(P_{\tau_1\otimes\tau_2} A^T P_{\tau_1\otimes\tau_2} \Lambda^{\otimes 2})/(|P_{\tau_1}||P_{\tau_2}|)$, with $P_{\tau_1}$ and $P_{\tau_2}$ being the projectors onto the image of $\tau_1$ and $\tau_2$ inside $\omega$, respectively.
\end{proof}

\section{Proof of Corollary~\ref{th:exdec3}}\label{appendssc:thm2proof}
{\bf Corollary 1.} {\it The expectation value $k_A(m)$ of the second-order correlation function defined in Eq.~\eqref{eq:otoccor} has an exponential decay behaviour with respect to the sequence length $m$, $k_A(m)=a (dO(t))^{m-1}$. The decay parameter is $p(A)=dO(t)$ with $O(t)=\frac{1}{d}\Tr(W^{\dagger}V(t)^{\dagger}WV(t))$ being OTOC.}

\begin{proof}
From Theorem~\ref{th:exdec1}, the function $k_A(m)$ has an exponential decay with decay parameter as 
$p(A) = \frac{\Tr(P_{ad}^{\otimes 2}A^{T} P_{ad}^{\otimes 2}\mathcal{U}_t^{\otimes 2})}{|P_{ad}|^2}$ and $P_{ad}$ is the corresponding projector onto the representation $\tau_{ad}$ inside $\omega$. Hence, we only need to prove that $p(A) = dO(t)$.

The projector $P_{ad}= \sum_{\sigma\in \mathbb{P}_n^0}\lketbra{\sigma}{\sigma}$ and $|P_{ad}|=d^2-1$. $P_{ad}^{\otimes 2}= \sum_{\sigma_1,\sigma_2 \in \mathbb{P}_n^0}\lketbra{\sigma_1\otimes \sigma_2}{\sigma_1\otimes \sigma_2}$. For two Pauli operators, there is $\Tr(\sigma_1\sigma_2)=\delta_{1,2}$. Note that $A^T=(d^2-1)^2\sum_{\sigma\in \mathbb{P}_n^0}\Tr(W \sigma W\sigma)\lketbra{V \otimes V}{\sigma\otimes \sigma}$ where $W$ and $V$ are both non-identity Pauli operators. Thus, operator $A^T$ is supported on the projection space of $P_{ad}^{\otimes2}$,
\begin{equation}
P_{ad}^{\otimes 2} A^T P_{ad}^{\otimes 2}=A^T.
\end{equation}
The decay parameter now becomes $p(A)=\frac{1}{(d^2-1)^2}\Tr(A^T\mathcal{U}_t^{\otimes 2})$. Before calculating $p(A)$, we first prove that in the expansion, $\lket{F(W \otimes W)}=\sum_{\sigma\in \mathbb{P}_n}\Tr(W \sigma W\sigma)\lket{\sigma \otimes\sigma}$ with $F$ being the swap operator.
In fact,
\begin{equation}
\begin{split}
\lbraket{\sigma_1 \otimes\sigma_2}{F(W\otimes W)}&=\Tr((\sigma_{1} \otimes\sigma_2)F(W\otimes W))\\
&=\Tr(F(W \sigma_1 \otimes W\sigma_2))\\
&=\Tr(W \sigma_1 W\sigma_2)\\
&=\Tr(W \sigma_1 W\sigma_1)\delta_{1,2}
\end{split}
\end{equation}
Since $F$ and $W$ are Hermitian matrices, the operator $A^T$ can be written as 
\begin{equation}
    A^T=(d^2-1)^2\left(\lketbra{V\otimes V}{F(W\otimes W)}-\lketbra{V\otimes V}{\mathbf{1}/\sqrt{d}\otimes\mathbf{1}/\sqrt{d}}\right).
\end{equation}
Therefore,
\begin{equation}
\begin{split}
\frac{\Tr(A^T\mathcal{U}_t^{\otimes 2})}{(d^2-1)^2}&=\Tr(\lketbra{V\otimes V}{F(W\otimes W)}\mathcal{U}_t^{\otimes 2})-\Tr(\lketbra{V\otimes V}{\mathbf{1}\otimes\mathbf{1}}\mathcal{U}_t^{\otimes 2})/d\\
 &=\lbraket{F(W\otimes W)}{\mathcal{U}_t (V)\otimes \mathcal{U}_t (V)}-\lbraket{\mathbf{1}\otimes\mathbf{1}}{\mathcal{U}_t (V)\otimes \mathcal{U}_t (V)}/d\\
 &=\Tr[F(W\mathcal{U}_t (V)\otimes W \mathcal{U}_t (V))]-|\Tr V|^{2}/d\\
 &=\Tr(W\mathcal{U}_t (V) W \mathcal{U}_t (V))=d\cdot O(t).
\end{split}
\end{equation}
Then, $k_A(t)$ has exponential decay with decay parameter as $p(A) = dO(t)$. 
\end{proof}

\section{Sample complexity}\label{appendssc:var}
In this section, we give a bound of the variance in OTOC measurement protocol, which also characterizes the sampling complexity. Recall that in the measurement protocol, $B_{xy}= \lketbra{\rho}{E_{x}}\otimes \lketbra{\rho}{E_{y}}$ and 
\begin{equation}\label{Eq:A}
\begin{split}
A&=(d^2-1)^2\sum_{\sigma\in \mathbb{P}_n^0}\Tr(W \sigma W \sigma)\lketbra{\sigma\otimes \sigma}{V
\otimes V}\\
&= (d^2-1)^2(\lketbra{F(W\otimes W)}{V\otimes V}-\lketbra{\mathbf{1}/\sqrt{d}\otimes\mathbf{1}/\sqrt{d}}{V\otimes V}),
\end{split}
\end{equation}
which we can directly obtain that $P_{ad}^{\otimes 2}AP_{ad}^{\otimes 2}=A$. Define the projector
\begin{equation}
P=\underset{g\in \mathbb{G}}{\mathbb{E}}\tau_{ad}(g)^{\otimes 2}\otimes \omega(g) = P_{ad}^{\otimes 2}\otimes \mathbf{1}(\underset{g\in \mathbb{G}}{\mathbb{E}}\omega(g)^{\otimes 3})P_{ad}^{\otimes 2}\otimes \mathbf{1}.
\end{equation}
Since the Clifford group forms a $3$-design \cite{PhysRevA.96.062336}, there is
\begin{equation}
    P =\sum_{\pi,\pi '\in S_3}Q_{\pi,\pi '}P_{ad}^{\otimes 2}\otimes \mathbf{1}\lketbra{\pi}{\pi '}P_{ad}^{\otimes 2}\otimes \mathbf{1},
\end{equation}
where $Q=(Q_{\pi,\pi '})$ is the Weingarten matrix, and $S_3$ is the permutation group for three copies of the base Hilbert space, i.e.,
\begin{equation}
\pi\ket{i_1,i_2,i_3}=\ket{i_{\pi(1)},i_{\pi(2)},i_{\pi(3)}}.
\end{equation}

\begin{theorem}
 Viewing $f_A$ defined in Eq.~\eqref{eq:otoccor} as a random variable with probability distribution from the shadow data, $p(x, \mathbf{g}^1)p(y, \mathbf{g}^2)$, the variance of this random variable is upper bounded by $O(d^{8m-12})$.
\end{theorem}
\begin{proof}
\begin{equation}
\begin{split}
\mathbf{Var}(f_A(m)) &=\mathbb{E}(f_A(m))^2-(\mathbb{E} f_A(m))^2\\
&\leq \sum_{x,y}\underset{\mathbf{g}^{1},\mathbf{g}^{2}\in \mathbb{G}^{\times m}}{\mathbb{E}} \Bigg[\Tr(B_{xy}[\tau_1 (g^{1}_{m})\otimes \tau_2(g^2_m)] \prod_{i=1}^{m-1} A[\tau_1(g^1_i)\otimes \tau_2(g^2_i)])\Bigg]^2 p(x,\mathbf{g}^{1})p(y,\mathbf{g}^{2})\\
&=\sum_{x,y}\Tr((B_{xy}^{\otimes 2}\otimes \lketbra{\Lambda_R(\tilde{\rho})}{\Lambda_L^*(\tilde{E_x})}\otimes \lketbra{\Lambda_R(\tilde{\rho})}{\Lambda_L^*(\tilde{E_y})}) P^{\otimes 2}\Big([A^{\otimes 2} \otimes \Lambda^{\otimes 2}] P^{\otimes 2} \Big)^{m-1})\\
&=\sum_{x,y}\Tr(\Theta^{xy}(Q^{\otimes 2}\Omega)^{m-1}Q^{\otimes 2}),
\end{split}
\end{equation}
where $\Theta^{xy}$ is the corresponding matrix with element
\begin{equation}
    \Theta^{xy}_{\pi,\pi ',\mu,\mu '}=\lbra{\pi '}\Big[P_{ad}^{\otimes 2}\lketbra{\rho}{E_{x}}^{\otimes 2}P_{ad}^{\otimes 2}\otimes \lketbra{\Lambda_R(\tilde{\rho})}{\Lambda_L^*(\tilde{E_x})}\Big]\lket{\pi}\lbra{\mu '}\Big[P_{ad}^{\otimes 2}\lketbra{\rho}{E_{y}}^{\otimes2}P_{ad}^{\otimes 2}\otimes \lketbra{\Lambda_R(\tilde{\rho})}{\Lambda_L^*(\tilde{E_y})}\Big]\lket{\mu},
\end{equation}
and $\Omega$ is the corresponding matrix with element
\begin{equation}\label{eq:Ome}
    \Omega_{\pi,\pi ',\mu,\mu '}=\lbra{\pi '\otimes \mu '}A_{14}\otimes A_{25}\otimes \Lambda_{3}\otimes\Lambda_{6}\lket{\pi\otimes\mu}.
\end{equation}
Here, there are six copies of the base Hilbert space with $\pi,\pi '$ being permutation operators of copies $123$ and $\mu,\mu '$ being permutation operators of copies $456$. In addition, the subscripts represent the copies that the operator acts on. For example, $A_{14}$ is an operator on the first and the fourth copies of the Hilbert space. 

Denote $\norm{A_{m\times n}}= \max(m,n) \underset{i,j}{\max}\abs{A_{ij}}$ and obviously,
\begin{equation}
    \Tr A_{m\times n}\leq \norm{A_{m\times n}}.
\end{equation}
Suppose $B$ is a $n\times s$ matrix, then 
\begin{equation}
\begin{split}
    \norm{AB}&=\max(m,s)\underset{i,k}{\max}\abs{\sum_{j}A_{ij}B_{jk}}\\
    &\leq n\max(m,s)\underset{i,j}{\max}\abs{A_{ij}} \underset{j,k}{\max}\abs{B_{jk}} \\
    &\leq \max(m,n)\max(n,s)\underset{i,j}{\max}\abs{A_{ij}} \underset{j,k}{\max}\abs{B_{jk}} =\norm{A}\norm{B}.
\end{split}
\end{equation}
Then, the variance can be further bounded by the inequality
\begin{equation}
    \begin{split}
        \mathbb{V}_{f_A (m)}&\leq \norm{\sum_{x,y}\Theta^{xy}(Q^{\otimes 2}\Omega)^{m-1}Q^{\otimes 2}}\\
        &\leq \norm{\Theta}\norm{Q}^{2m}\norm{\Omega}^{m-1}\\
        &\leq c(\underset{i,j}{\max}\abs{\Theta_{ij}})(\underset{i,j}{\max}\abs{Q_{ij}})^{2m}(\underset{i,j}{\max}\abs{\Omega_{ij}})^{m-1}\\
        &=c(\frac{d^2-2}{d(d^2-1)(d^2-4)})^{2m}(\underset{i,j}{\max}\abs{\Theta_{ij}})(\underset{i,j}{\max}\abs{\Omega_{ij}})^{m-1}\\
        &=O(d^{-6m}) (\underset{i,j}{\max}\abs{\Theta_{ij}})(\underset{i,j}{\max}\abs{\Omega_{ij}})^{m-1},
    \end{split}
\end{equation}
where $c$ is a constant independent of the system dimension $d$ and $\Theta=\sum_{x,y}\Theta_{xy}$. Thus, we only need to calculate the maximal element of matrices $\Theta$ and $\Omega$.


Notice that $\Theta=\tilde{\Theta}^{\otimes 2}$ where the matrix $\tilde{\Theta}$ has the element 
\begin{equation}\label{eq:apptheta}
   \tilde{\Theta}_{\pi,\pi'} =\sum_{x}\lbra{\pi '}\Big[P_{ad}^{\otimes 2}\lketbra{\rho}{E_{x}}^{\otimes 2}P_{ad}^{\otimes 2}\otimes \lketbra{\Lambda_R(\tilde{\rho})}{\Lambda_L^*(\tilde{E_x})}\Big]\lket{\pi}.
\end{equation}
The maximal element of matrix $\tilde{\Theta}$ is in the order of $d$ \cite{helsen2021estimating}. Thus, $\underset{i,j}{\max}\abs{\Theta_{ij}}=O(d^2)$.

To calculate the matrix $\Omega$, we notice that the projector $P_{ad}$ is supported on the space of traceless matrices, that is,
\begin{equation}
    \begin{split}
       P_{ad}\lket{\mathbf{1}}&=\sum_{\sigma\in\mathbb{P}_n^0}\lket{\sigma}\lbraket{\sigma}{\mathbf{1}}=0,\\
       \lbra{\mathbf{1}}P_{ad}&=\sum_{\sigma\in\mathbb{P}_n^0}\lbraket{\mathbf{1}}{\sigma}\lbra{\sigma} =0.
    \end{split}
\end{equation}
The construction of $A$ satisfies that $A=P_{ad}^{\otimes 2}A P_{ad}^{\otimes 2}$. Thus, there are four projectors $P_{ad}$ acting on copies $1245$ within Eq.~\eqref{eq:Ome}. If one of the copies $1245$ remains unchanged under a permutation, then the corresponding matrix elements of $\Omega$ are zero. For example, we take $\pi=(13)$ and then
\begin{equation}
    P_{ad}^{\otimes 4}(A_{14}\otimes A_{25})P_{ad}^{\otimes 4}\otimes\Lambda_3\otimes\Lambda_6\lket{(13)\otimes\mathbf{1}_2\otimes\mu }=0.
\end{equation}
Therefore, the element $\Omega_{\pi,\pi ',\mu,\mu '}$ is nonzero iff $\pi,\pi '\in \{(12),(123),(132)\}$ and $\mu,\mu '\in \{(45),(456),(465)\}$. Next, we leverage the tool of the tensor network to calculate the element of matrix $\Omega$. From Eq.\eqref{Eq:A}, there is $A^{\otimes 2}=(d^2-1)^4(B+C/d^2-D/d-E/d)$ with
\begin{equation}
\begin{split}
    B&=\lketbra{F(W\otimes W)^{\otimes 2}}{V^{\otimes 4}},\\
    C&=\lketbra{\mathbf{1}^{\otimes 4}}{V^{\otimes 4}},\\
    D&=\lketbra{F(W\otimes W)\otimes \mathbf{1}^{\otimes 2}}{V^{\otimes 4}},\\
    E&=\lketbra{\mathbf{1}^{\otimes 2}\otimes F(W\otimes W)}{V^{\otimes 4}}.  
\end{split}
\end{equation}
To bound the maximum element of $\Omega$, we need to calculate the maximal element of the matrices generated by $B$, $C$, and $D$, respectively. By symmetry, the elements of $D$ and $E$ have the same order.

For $B$, we first take $\pi=(123),\pi '=(12),\mu=\mu '=(456)$ as an example. There is
\begin{equation}
    \lbra{(12)\otimes(456)}B\otimes \Lambda_{3}\otimes\Lambda_{6}\lket{(123)\otimes(456)}=\begin{tabular}{c}
     \includegraphics[scale=0.2]{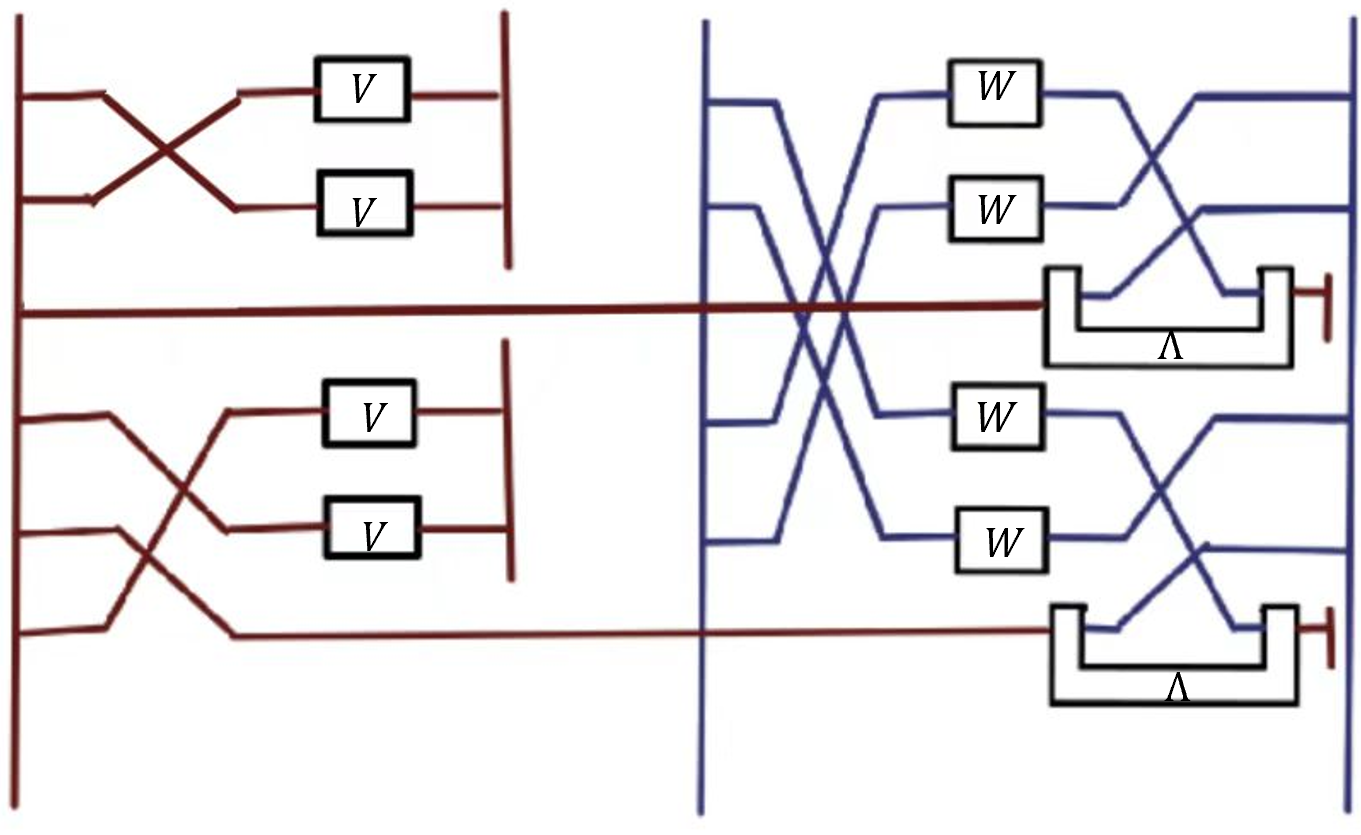}
\end{tabular}
=
d \left(\begin{tabular}{c}
\includegraphics[scale=0.3]{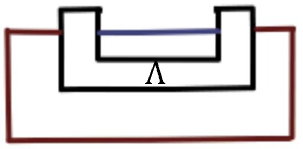}
\end{tabular}\right)^2=d[\Tr(\Lambda(\mathbf{1}))]^2=d^3.
\end{equation}
Similarly, 
\begin{equation}
    \lbra{(123)\otimes(456)}B\otimes \Lambda_{3}\otimes\Lambda_{6}\lket{(12)\otimes(45)}=\begin{tabular}{c}
     \includegraphics[scale=0.2]{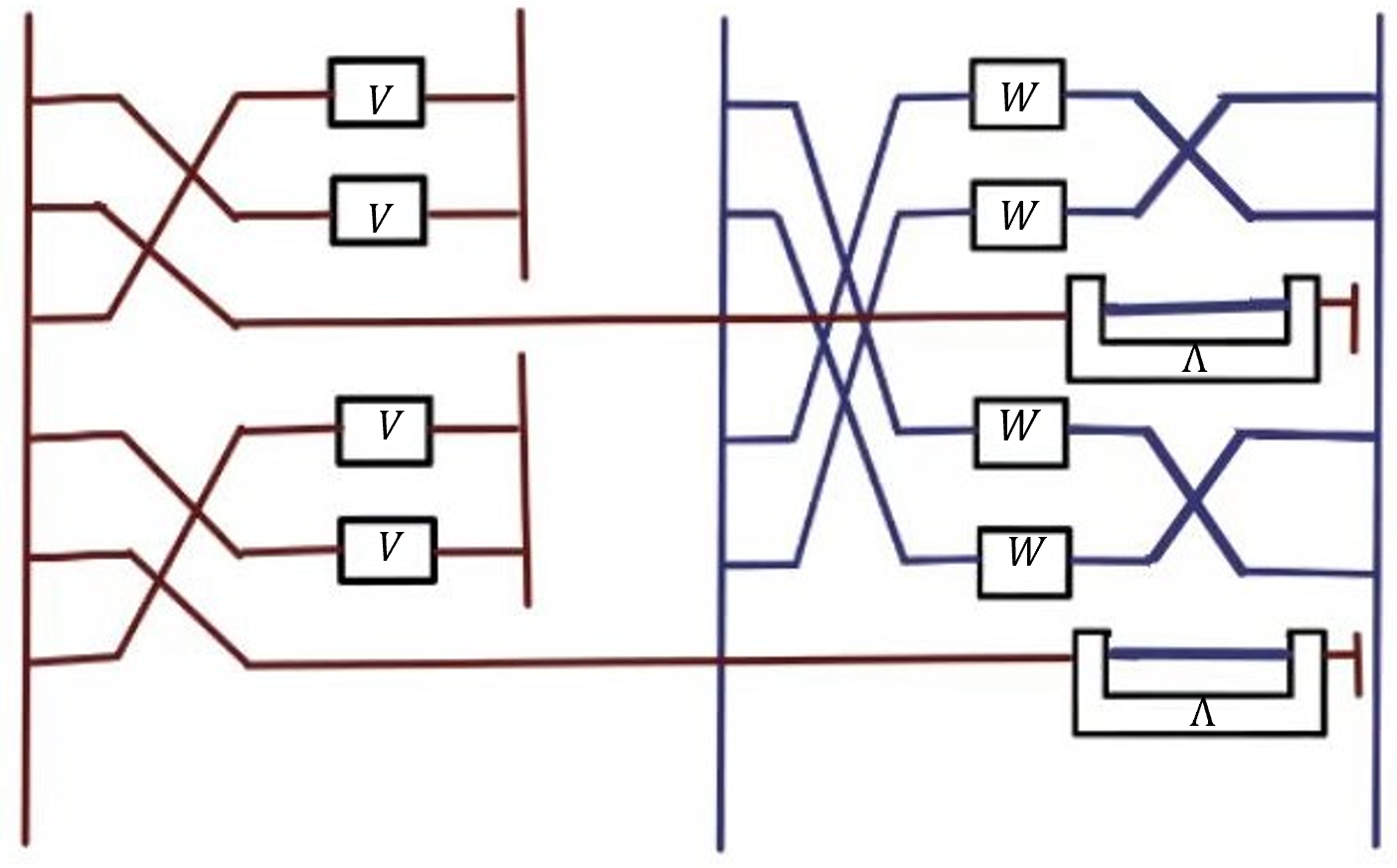}
\end{tabular}
=d^2
\left(\begin{tabular}{c}
     \includegraphics[scale=0.3]{figure3/d.png}
\end{tabular}\right)^2=d^4.
\end{equation}
From the symmetry of the system and these two examples above, we found that the maximum element corresponding to $B$ is obtained when $\pi=\pi'=(12),\mu=\mu'=(45)$. In this case, there is
\begin{equation}
\lbra{(12)\otimes(45)}B\otimes \Lambda_{3}\otimes\Lambda_{6}\lket{(12)\otimes(45)} = \begin{tabular}{c}
\includegraphics[scale=0.36]{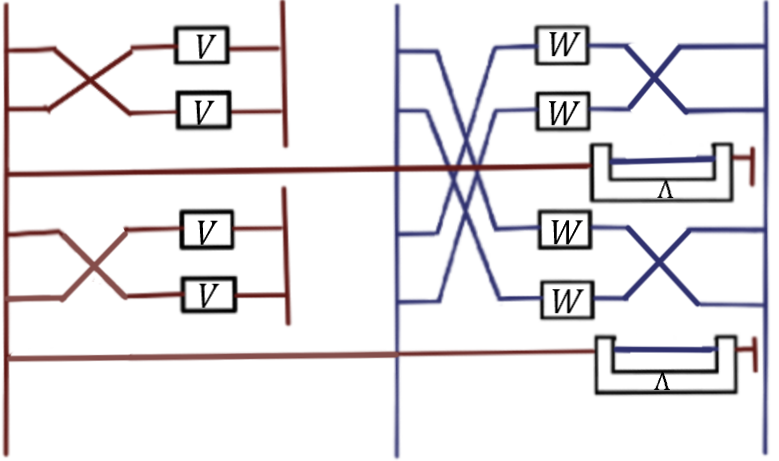}
\end{tabular}=\Tr(V^2)^2\Tr(W^2)^2\Tr(\Lambda(\mathbf{1}))^2 = d^6.
\end{equation}

Similarly, we can obtain for $C$,
\begin{equation}
   \underset{\pi,\mu,\pi',\mu'}{\max} \lbra{\pi ',\mu '}C\otimes \Lambda_{3}\otimes\Lambda_{6}\lket{\pi,\mu}= \lbra{(12)\otimes(45)}C\otimes \Lambda_{3}\otimes\Lambda_{6}\lket{(12)\otimes(45)} = d^6.
\end{equation}
and for $D$,
\begin{equation}
\underset{\pi,\mu,\pi',\mu'}{\max} \lbra{\pi ',\mu '}D\otimes \Lambda_{3}\otimes\Lambda_{6}\lket{\pi,\mu} = \lbra{(12)\otimes(45)}D\otimes \Lambda_{3}\otimes\Lambda_{6}\lket{(12)\otimes(45)} = d^5.
\end{equation}
Consequently, 
\begin{equation}\label{eq:appvar}
 \mathbf{Var}(f_A(m))\leq O(d^{8m-12}).
\end{equation}

\end{proof}

From the variance of the single-shot experiment, we first calculate the variance of the estimator given in Eq.~\eqref{eq:kfmhat}.

\begin{corollary}
   Suppose the number of samples is $S$ and the estimator for $k_A(m,S)$ is 
    \begin{equation}
\hat{k}_A(m,S) = \frac{1}{S(S-1)}\sum_{i\neq j}f(x_i,x_j,\mathbf{g}^i,\mathbf{g}^j, m).
\end{equation}
Then the variance has an upper bound,
\begin{equation}\label{eq:appsamp}
    \mathbf{Var}[\hat{k}_A(m,S)]\leq \left\{
    \begin{aligned}
    \frac{1}{S(S-1)}O(d^{-4})+\frac{2(S-2)}{S(S-1)}O(d^{-1}),\quad m=1\\
    \frac{1}{S(S-1)}O(d^{4})+\frac{2(S-2)}{S(S-1)}O(1), \quad m=2\\
    \frac{1}{S(S-1)}O(d^{8m-12})+\frac{2(S-2)}{S(S-1)}O(d^{4m-4}),\quad m> 2
    \end{aligned}
    \right
    .
\end{equation}
\end{corollary}
\begin{proof}
 Denote $s_i=(x_i,\mathbf{g}^i)$ and the variance is
\begin{equation}\label{eq:appvarsamp}
    \begin{split}
\mathbf{Var}[\hat{k}_A(m,S)]&=\frac{1}{S^2(S-1)^2} \sum_{i\neq j,i'\neq j'}\mathbf{Cov}[f(s_i,s_j),f(s_{i'},s_{j'})]\\
&=\frac{1}{S^2(S-1)^2}\Big(\sum_{i=i',j=j',i\neq j}\mathbf{Cov}[f(s_i,s_j),f(s_{i'},s_{j'})] +\sum_{i=i',i\neq j,j\neq j',i\neq j'}\mathbf{Cov}[f(s_i,s_j),f(s_{i'},s_{j'})]\\
&+\sum_{j=j',i\neq j,i\neq i',i'\neq j}\mathbf{Cov}[f(s_i,s_j),f(s_{i'},s_{j'})]\Big)\\
&=\frac{1}{S^2(S-1)^2}\Big(\sum_{i\neq j}\mathbf{Var}[f(s_i,s_j)]+2\sum_{i\neq j,j\neq j',i\neq j'}\mathbf{Cov}[f(s_i,s_j),f(s_{i},s_{j'})]\Big)\\
&=\frac{1}{S(S-1)}\mathbf{Var}[f(s_1,s_2)]+\frac{2(S-2)}{S(S-1)}\mathbf{Cov}[f(s_1,s_2),f(s_1,s_3)]\\
&=\frac{1}{S(S-1)}\mathbf{Var}[f_A(m)]+\frac{2(S-2)}{S(S-1)}\mathbf{Cov}[f(s_1,s_2),f(s_1,s_3)].
    \end{split}
\end{equation}
The first term has already been given, and we only need to calculate the second term. When the sample $s_1$ is taken as a fixed data, $f(s_1,s_2)$ and $f(s_1,s_3)$ are independent. Therefore,
\begin{equation}
    \begin{split}
        \mathbf{Cov}[f(s_1,s_2),f(s_1,s_3)]&=\underset{s_1,s_2,s_3}{\mathbb{E}}[f(s_1,s_2)-k_A(m,S)][f(s_1,s_3)-k_A(m,S)]\\
        &=\underset{s_1}{\mathbb{E}} \Big( \underset{s_2,s_3}{\mathbb{E}}\Big[(f(s_1,s_2)-k_A(m,S))(f(s_1,s_3)-k_A(m,S)) \Big| s_1\Big] \Big)\\
      &=\underset{s_1}{\mathbb{E}}\Big( (\underset{s_2}{\mathbb{E}}[f(s_1,s_2)|s_1]-k_A(m,S))(\underset{s_3}{\mathbb{E}}[f(s_1,s_3)|s_1]-k_A(m,S))\Big)\\
      &=\mathbf{Var}[f(s_1)],
    \end{split}
\end{equation}
where $f(s_1)=\underset{s_2}{\mathbb{E}}[f(s_1,s_2)|s_1]$ is the conditional expectation when sample value $s_1$ is given. Explicitly,
\begin{equation}\label{eq:appconexp}
    \begin{split}
        f(s_1)&=\sum_{y}\underset{\mathbf{g}^{2}\in \mathbb{G}^{\times m}}{\mathbb{E}} \Tr(B_{xy}[\tau_1 (g^{1}_{m})\otimes \tau_2(g^2_m)] \prod_{i=1}^{m-1} A[\tau_1(g^1_i)\otimes \tau_2(g^2_i)]) p(y,\mathbf{g}^{2})\\
        &=\sum_{y}\Tr[(B_{xy}\otimes \lketbra{\Lambda_R(\tilde{\rho})}{\Lambda_L^*(\tilde{E_y})})(\tau_1(g^1_m) \otimes P')\prod_{i=1}^{m-1}(\mathbf{1}\otimes P')(A\otimes \Lambda)(\tau_1(g^1_i)\otimes P') ],
    \end{split}
\end{equation}
where 
\begin{equation}
    \begin{split}
        P'&=P_{ad}\otimes\mathbf{1}(\underset{g\in\mathbb{G}}{\mathbb{E}}\omega(g)^{\otimes 2})P_{ad}\otimes\mathbf{1}\\
        &=P_{ad}\otimes\mathbf{1}\Big(\frac{1}{d^2}\lketbra{\mathbf{1}^{\otimes 2}}{\mathbf{1}^{\otimes 2}} +\frac{1}{d^2-1}\sum_{\sigma,\sigma'\in\mathbb{P}_n^0} \lketbra{\sigma^{\otimes 2}}{\sigma'^{\otimes 2}}\Big)P_{ad}\otimes\mathbf{1}\\
        &=\frac{1}{d^2-1}\sum_{\sigma,\sigma'\in\mathbb{P}_n^0} \lketbra{\sigma^{\otimes 2}}{\sigma'^{\otimes 2}}\\
        &=\frac{1}{d^2-1}\lketbra{F-\mathbf{1}^{\otimes 2}/d}{F-\mathbf{1}^{\otimes 2}/d}=\lketbra{B_2}{B_2}.
    \end{split}
\end{equation}
Here, we use the fact that the Clifford group forms a two-design \cite{helsen2021estimating}.

To calculate the variance of $f(s_1)$, we need to simplify Eq.\eqref{eq:appconexp}. Specifically,
\begin{equation}
    \begin{split}
        &(\mathbf{1}\otimes P')(A\otimes \Lambda)(\tau_1(g^1_i)\otimes P')\\
        &= \sum_{\alpha,\alpha',\beta,\beta',\sigma\in\mathbb{P}_n^0} \Tr(W \sigma W\sigma)(\mathbf{1}\otimes\lketbra{\alpha^{\otimes 2}}{\alpha'^{\otimes 2}})(\lketbra{\sigma\otimes \sigma}{V\otimes V}\otimes \Lambda)(\tau_1(g^1_i)\otimes\lketbra{\beta^{\otimes 2}}{\beta'^{\otimes 2}})\\
        &=\sum_{\alpha,\alpha',\beta,\beta',\sigma\in\mathbb{P}_n^0} \Tr(W \sigma W\sigma) \Tr(\sigma\alpha')\Tr(V\beta)\Tr(\alpha'\Lambda(\beta))\Big[\lketbra{\sigma}{V}\tau_1(g^1_i)\otimes \lketbra{\alpha}{\beta'}\otimes \lketbra{\alpha}{\beta'}\Big]\\
        &=\sum_{\alpha,\beta',\sigma\in\mathbb{P}_n^0} \Tr(W \sigma W \sigma) \Tr(\sigma\Lambda(V))\Big[\lketbra{\sigma}{V}\tau_1(g^1_i)\otimes \lketbra{\alpha}{\beta'}\otimes \lketbra{\alpha}{\beta'}\Big]\\
        &=\sum_{\sigma\in\mathbb{P}_n^0} (d^2-1) \Tr(W \sigma W \sigma) \Tr(\sigma\Lambda(V))\lketbra{\sigma}{V}\tau_1(g^1_i)\otimes P'=A'\tau_1(g^1_i)\otimes P',
    \end{split}
\end{equation}
where $A'=(d^2-1)\sum_{\sigma\in\mathbb{P}_n^0} \Tr(W \sigma W \sigma) \Tr(\sigma\Lambda(V))\lketbra{\sigma}{V}$. Therefore,
\begin{equation}
    \begin{split}
        f(s_1)  &=\sum_{y}\Tr( \Big(\lketbra{\rho}{E_x}\tau_1(g_m^1)\prod_{i=1}^{m-1}A'\tau_1(g^1_i)\Big)\otimes \Big(\lketbra{\rho\otimes\Lambda_R(\tilde{\rho})}{E_y\otimes\Lambda_L^*(\tilde{E_y})} P'\Big) ),\\
        &=\sum_y\frac{1}{d^2-1}\Big(\Tr[\rho\Lambda_R(\tilde{\rho})]-\frac{1}{d}\Tr[\Lambda_R(\tilde{\rho})]\Big)\Big(\Tr[E_y \Lambda_L^*(\tilde{E_y})]-\frac{1}{d}\Tr E_y\Tr[\Lambda_L^*(\tilde{E_y})]\Big)\\
    &\times \Tr(\lketbra{\rho}{E_x}\tau_1(g_m^1)\prod_{i=1}^{m-1}A'\tau_1(g^1_i))\\
    &=O(d^{-1})\Tr(\lketbra{\rho}{E_x}\tau_1(g_m^1)\prod_{i=1}^{m-1}A'\tau_1(g^1_i)).
    \end{split}    
\end{equation}

When $m=1$, there is 
\begin{equation}
    \begin{split}
        \mathbf{Var}[f(s_1)]&\leq O(d^{-2})\sum_x \underset{g\in \mathbb{G}}{\mathbb{E}}\lbra{E_x}\tau(g)\lket{\rho}^2 \lbra{\Lambda_L^*(\tilde{E_x})}\omega(g)\lket{\Lambda_R(\tilde{\rho})}\\
        &= O(d^{-2})\sum_x\lbra{E_x^{\otimes 2}\otimes\Lambda_L^*(\tilde{E_x})}P\lket{\rho^{\otimes 2}\otimes\Lambda_R(\tilde{\rho})}\\
        &=O(d^{-2})\sum_x \sum_{\pi\in S_3}Q_{\pi,\pi }\lbraket{(E_x P_{ad})^{\otimes 2}\otimes\Lambda_L^*(\tilde{E_x})}{\pi}\lbraket{\pi}{(P_{ad} \rho)^{\otimes 2}\otimes\Lambda_R(\tilde{\rho})}\\
        &=O(d^{-2})\Tr\tilde{\Theta}=O(d^{-1}),
    \end{split}
\end{equation}
where the matrix $\tilde{\Theta}$ is defined in Eq.\eqref{eq:apptheta} and the elements of this matrix have the order $O(d)$.

When $m\geq 2$, according to the conclusion in Ref.\cite{helsen2021estimating}, the variance for $f(s_1)$ has an upper bound,
\begin{equation}
    \mathbf{Var}[f(s_1)]\leq O(d^{-4})\Big[11 u(A')\Big(r(A')^{m-2}+[2(m-2)^2 r(A')^{m-3}] \max\{11 u(A'),(11u(A'))^2\}\Big)\Big],
\end{equation}
where $u(A')=\Tr(A'A'^{\dagger})/(d^2-1)$ and $r(A')=u(A')(1+16 d^{-1/3})$. Since
\begin{equation}
    \begin{split}
\Tr(A'A'^{\dagger})&=(d^2-1)^2\sum_{\sigma,\sigma'\in\mathbb{P}_n^0}\Tr(W\sigma W\sigma)\Tr(W\sigma'W\sigma')\Tr(\sigma\Lambda(V))\Tr(\sigma'\Lambda(V))\Tr(\sigma\sigma')\Tr(V^2)\\
&=(d^2-1)^2 d\sum_{\sigma\in\mathbb{P}_n^0}[\Tr(W\sigma W\sigma)\Tr(\sigma\Lambda(V))]^2\\
&=(d^2-1)^2 d\sum_{\sigma\in\mathbb{P}_n^0}\Tr(\sigma^{\otimes 2}\Lambda(V)^{\otimes 2})\\
&=(d^2-1)^2 d \Tr[ (F-\mathbf{1}^{\otimes 2}/d)\Lambda(V)^{\otimes 2} ]\\
&=(d^2-1)^2 d \left(\Tr(\Lambda(V)^2)-\frac{1}{d}[\Tr(\Lambda(V))]^2\right)\\
&=O(d^6),
    \end{split}
\end{equation}  
 we can obtain $u(A')=O(d^4)$ and $r(A')=O(d^4)$. When $m> 2$, the variance can be bounded by $O(d^{4m-4})$. When $m=2$, the variance can be bounded by $O(1)$, which implies that the variance is bounded by a constant. From Eq.\eqref{eq:appvarsamp}, the overall variance is
\begin{equation}
    \mathbf{Var}[\hat{k}_A(m,S)]\leq \left\{
    \begin{aligned}
    \frac{1}{S(S-1)}O(d^{-4})+\frac{2(S-2)}{S(S-1)}O(d^{-1}),\quad m=1\\
    \frac{1}{S(S-1)}O(d^{4})+\frac{2(S-2)}{S(S-1)}O(1), \quad m=2\\
    \frac{1}{S(S-1)}O(d^{8m-12})+\frac{2(S-2)}{S(S-1)}O(d^{4m-4}),\quad m> 2
    \end{aligned}
    \right
    .
\end{equation}


\end{proof}
Next, we calculate the variance of OTOC evaluated by the ratio of $k_A(2)$ and $k_A(1)$ as shown below.
\begin{equation}
\hat{O} =\frac{\hat{k}_A(2,S)}{\hat{k}_A(1,S)}=\frac{\frac{1}{S(S-1)}\sum_{i\neq j}f(x_i,x_j,\mathbf{g}^i,\mathbf{g}^j, 2)}{\frac{1}{S(S-1)}\sum_{i\neq j}f(x_i,x_j,\mathbf{g}^i,\mathbf{g}^j, 1)}
\end{equation}
where we assume that the numbers of samples for two cases are both equal to $S$. For an arbitrary $m$, the expectation value of $\hat{k}_A(m,S)$ is 
\begin{equation}
    \begin{split}
        \mathbb{E}\hat{k}_A(m,S)&=k_A(m,S)\\
        &=\sum_{x,y}\underset{\mathbf{g}^{1},\mathbf{g}^{2}\in \mathbb{G}^{\times m}}{\mathbb{E}} \Bigg[\Tr(B_{xy}[\tau_1 (g^{1}_{m})\otimes \tau_2(g^2_m)] \prod_{i=1}^{m-1} A[\tau_1(g^1_i)\otimes \tau_2(g^2_i)])\Bigg] p(x,\mathbf{g}^{1})p(y,\mathbf{g}^{2})\\
        &=\sum_{x,y}\Tr((B_{xy}\otimes \lketbra{\Lambda_R(\tilde{\rho})}{\Lambda_L^*(\tilde{E_x})}\otimes \lketbra{\Lambda_R(\tilde{\rho})}{\Lambda_L^*(\tilde{E_y})}) {P'}^{\otimes 2}\Big([A \otimes \Lambda^{\otimes 2}] P'^{\otimes 2} \Big)^{m-1}).
    \end{split}
\end{equation}

When $m=1$, there is 
\begin{equation}\label{eq:appave1}
\begin{split}
     k_A(1,S)&=\lbraket{B_2}{\rho\otimes\Lambda_R(\tilde{\rho})}^2(\sum_x \lbraket{E_x\otimes\Lambda_L^*(\tilde{E_x})}{B_2})^2\\
     &=O(d^{-4})(\Tr[\rho\Lambda_R(\tilde{\rho})]-\frac{1}{d}\Tr[\Lambda_R(\tilde{\rho})])^2\Big(\sum_x \Tr[E_x\Lambda_L^*(\tilde{E_x})] - \frac{1}{d}\sum_x\Tr E_x \Tr[\Lambda_L^*(\tilde{E_x})] \Big)^2= O(d^{-2}).
\end{split}
\end{equation}
When $m=2$, there is 
\begin{equation}\label{eq:appave2}
\begin{split}
    k_A(2,S)&=k_A(1,S)\lbra{B_2^{\otimes 2}}A\otimes \Lambda^{\otimes 2}\lket{B_2^{\otimes 2}}\\
    &=O(d^{-2})(\lbra{B_2^{\otimes 2}}[\lketbra{F(W\otimes W)}{V\otimes V}\otimes\Lambda^{\otimes 2}]\lket{B_2^{\otimes 2}}-\frac{1}{d}\lbra{B_2^{\otimes 2}}[\lketbra{\mathbf{1}\otimes\mathbf{1}}{V\otimes V}\otimes\Lambda^{\otimes 2}]\lket{B_2^{\otimes 2}})\\
    &=O(d^{-2})\lbra{B_2^{\otimes 2}}[\lketbra{F(W\otimes W)}{V\otimes V}\otimes\Lambda^{\otimes 2}]\lket{B_2^{\otimes 2}}\\
    &=O(d^{-2})\begin{tabular}{c}
     \includegraphics[scale=0.4]{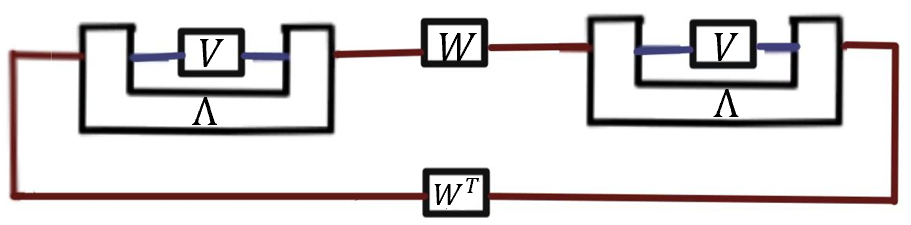}\end{tabular}\\
     &=O(d^{-2})\Tr(W\Lambda(V)W\Lambda(V))=O(d^{-1}).
\end{split}
\end{equation}

For two random variables $X,Y$, the uncertainty of the division $X/Y$ is
\begin{equation}
    \Delta\left(\frac{X}{Y}\right)\approx\left(\frac{\mu_X}{\mu_Y}\right)^2\left[\frac{\sigma_X^2}{\mu_X^2}-2\frac{\mathbf{Cov}(X,Y)}{\mu_X\mu_Y}+\frac{\sigma_Y^2}{\mu_Y^2}\right],
\end{equation}
where $\mu_X$ and $\mu_Y$ are the expectation values of $X$ and $Y$, respectively; $\sigma^2_X$ and $\sigma^2_Y$ are the variance values of $X$ and $Y$, respectively. In our case, $X=\hat{k}_A(2,S)$ and $Y=\hat{k}_A(1,S)$. According to Eq.~\eqref{eq:appsamp}, Eq.~\eqref{eq:appave1}, and Eq.~\eqref{eq:appave2}, there is 
\begin{equation}
    \mu_X^2=O(d^{-2}), \mu_Y^2=O(d^{-4}), \sigma_X^2=O(d^{4}S^{-2})+ O(S^{-1}),\sigma_Y^2=O(d^{-4}S^{-2})+O(d^{-1}S^{-1}).
\end{equation}

Thus, the uncertainty for the estimator can be bounded by
\begin{equation}
    \mathrm{Var}\left(\frac{\hat{k}_A(2,S)}{\hat{k}_A(1,S)}\right)=\mathrm{Var}\left(\frac{X}{Y}\right)\leq \left(\frac{\mu_X}{\mu_Y}\right)^2\left[\frac{\sigma_X^2}{\mu_X^2}+\frac{\sigma_Y^2}{\mu_Y^2}\right] = O\left(\frac{d^{8}}{S^{2}}\right)+O\left(\frac{d^{5}}{S}\right).
\end{equation}

\section{Proof of Theorem \ref{th:exdec2}}\label{appenssc:thm3proof}
{\bf Theorem 2.} {\it (Exponential decay) Given a generalized UIRS protocol with respect to the group $\mathbb{G}$, the expectation value $k_A(m)$ defined with Eqs.~\eqref{eq:expec} and~\eqref{eq:groupfunction} has the following form:
\begin{equation}
    k_{A}(m)=\Tr(\Theta(\{E_x\},\rho)[\Phi(A,\Lambda)]^{m-1}),
\end{equation}
where $\Theta(\{E_x\},\rho)$ and $\Phi(A,\Lambda)$ are induced matrices dependent on SPAM and channel $\Lambda$, respectively. The matrix $\Phi(A,\Lambda)$ has the element
\begin{equation}\label{eq:decayrandom}
    [\Phi(A,\Lambda)]_{i,j}=\frac{1}{|P_j|}\Tr(P_i A^{T}P_{j} \Lambda^{\otimes 2}),
\end{equation}
where $P_i$ and $P_j$ are the projectors onto the $i$-th and $j$-th copies of $\tau$ inside $\omega^{\otimes 2}$, respectively.
In the special case that the representation $\omega^{\otimes 2}$ only has one copy of $\tau$, the matrix $\Phi(A,\Lambda)$ reduces into a number.}

\begin{proof}
\begin{equation}
\begin{split}
k_A(m)&=\sum_{x,y\in\mathcal{X}}\underset{\mathbf{g}\in \mathbb{G}^{\times m}}{\mathbb{E}}\Tr(B_{xy} \tau(g_{m})\prod_{i=1}^{m-1}A\tau(g_i))\lbra{\tilde{E_x}}\prod_{i=1}^{m}\Lambda_L\omega(g_i)\Lambda_R\lket{\tilde{\rho}}\lbra{\tilde{E_y}}\prod_{i=1}^{m}\Lambda_L\omega(g_i)\Lambda_R\lket{\tilde{\rho}}\\
&=\sum_{x,y\in\mathcal{X}}\underset{\mathbf{g}\in \mathbb{G}^{\times m}}{\mathbb{E}}
\Tr((B_{xy}\otimes \lketbra{\Lambda_R(\tilde{\rho})}{\Lambda_L^*(\tilde{E_x})}\otimes \lketbra{\Lambda_R(\tilde{\rho})}{\Lambda_L^*(\tilde{E_y})})\tau(g_m)\otimes\omega(g_m)^{\otimes 2}\prod_{i=1}^{m-1}(A\otimes \Lambda^{\otimes 2}(\tau(g_i)\otimes\omega(g_i)^{\otimes 2})))\\
&=\sum_{x,y\in\mathcal{X}}\Tr((B_{xy} \otimes \lketbra{\Lambda_R(\tilde{\rho})}{\Lambda_L^*(\tilde{E_x})}\otimes \lketbra{\Lambda_R(\tilde{\rho})}{\Lambda_L^*(\tilde{E_y})})\Big(\underset{g\in \mathbb{G}}{\mathbb{E}}(\tau(g)\otimes \omega(g)^{\otimes 2})(A\otimes \Lambda^{\otimes 2}) \underset{g\in\mathbb{G}}{\mathbb{E}}(\tau(g)\otimes \omega(g)^{\otimes 2})\Big)^{m-1} )\\
&=\sum_{x,y\in\mathcal{X}}\Tr((B_{xy} \otimes \lketbra{\Lambda_R(\tilde{\rho})}{\Lambda_L^*(\tilde{E_x})}\otimes \lketbra{\Lambda_R(\tilde{\rho})}{\Lambda_L^*(\tilde{E_y})})\Big(P_{\tau}(A\otimes \Lambda^{\otimes 2})P_{\tau}\Big)^{m-1}),
\end{split}
\end{equation}
where $P_{\tau}=\underset{g\in \mathbb{G}}{\mathbb{E}}\tau(g)\otimes \omega(g)^{\otimes 2}$ is a projector. The decay parameter $P_{\tau}(A\otimes \Lambda^{\otimes 2})P_{\tau}$ can be viewed as a matrix, $\Phi$, with dimension $\rank(P_{\tau})\times \rank(P_{\tau})$ when restricting in the space associated with projector $P_{\tau}$. The matrix $\Phi$ has the element
\begin{equation}
    [\Phi]_{i,j}=\frac{1}{|P_j|}\Tr(P_i A^{T}P_{j} \Lambda^{\otimes 2}),
\end{equation}
where $P_i$ is the projector onto the $i$-th copy of $\tau$ inside $\omega^{\otimes 2}$.
\end{proof}


\section{Formula simplification of Eq.~\eqref{eq:otoccor}}\label{appendixssc:formula}

In this section, we simplify the computation of Eq.~\eqref{eq:otoccor}:
\begin{equation}
    f_A(x,y,\mathbf{g^1},\mathbf{g^2}, m)=\Tr\left(B_{xy}\tau(g_m^1)\otimes\tau(g_m^2)\prod_{i=1}^{m-1}A\tau(g_i^1)\otimes\tau(g_i^2)\right).
\end{equation}
For the Clifford group element $g$, with slight abuse of notations, we denote $g=\omega(g)=U_g$ and $\tau_{ad}(g)=\tau(g)$. Then $g = \tau(g)+\frac{1}{d}\lket{\mathbf{1}}\lbra{\mathbf{1}}$. When $m=1$, we have
\begin{equation}\label{eq:f_A1}
    \begin{split}
    f_A(x,y,\mathbf{g^1},\mathbf{g^2}, m)=&\Tr\left(B_{xy}\tau(g_m^1)\otimes\tau(g_m^2)\right)\\
    =&\Tr\left(\lket{\rho}\lbra{E_x}\otimes\lket{\rho}\lbra{E_y}\tau(g^1)\otimes\tau(g^2)\right)\\
    =&\Tr(\lket{\rho}\lbra{E_x}\tau(g^1))\cdot \Tr(\lket{\rho}\lbra{E_y}\tau(g^2))\\
    =&\lbra{E_x}\tau(g^1)\lket{\rho}\cdot\lbra{E_y}\tau(g^2)\lket{\rho}\\
    =&\lbra{E_x}(\omega(g^1)-\frac{1}{d}\lket{\mathbf{1}}\lbra{\mathbf{1}})\lket{\rho}\cdot\lbra{E_y}(\omega(g^2)-\frac{1}{d}\lket{\mathbf{1}}\lbra{\mathbf{1}})\lket{\rho}\\
    =&\left(\Tr(E_xg^1\rho g^{1\dagger})-\frac{1}{d}\right)\cdot\left(\Tr(E_yg^2\rho g^{2\dagger})-\frac{1}{d}\right).
    \end{split}
\end{equation}

For $m\ge 2$, $A=(d^2-1)^2\sum_{\sigma\in \mathbb P_n^0}\Tr(W \sigma W \sigma)\lketbra{\sigma\otimes \sigma}{V \otimes V}$,
\begin{equation}\label{eq:b6}
\begin{split}
    f_A(x,y,\mathbf{g^1},\mathbf{g^2}, m)=&\Tr\left(B_{xy}\tau(g_m^1)\otimes\tau(g_m^2)\prod_{i=1}^{m-1}A\tau(g_i^1)\otimes\tau(g_i^2)\right)\\
    =&\Tr\left(B_{xy}\tau(g_m^1)\otimes\tau(g_m^2)\prod_{i=1}^{m-1}(d^2-1)^2\sum_{\sigma\in \mathbb P_n^0}\Tr(W \sigma W \sigma
    )\lketbra{\sigma\otimes \sigma}{V\otimes V}\tau(g_i^1)\otimes\tau(g_i^2)\right)\\
    =&(d^2-1)^{2(m-1)}\cdot \Tr \left[B_{xy}\tau(g_m^1)\otimes\tau(g_m^2)\sum_{\sigma\in \mathbb P_n^0}\Tr(W \sigma W\sigma)\lket{\sigma\otimes\sigma}\lbra{V\otimes V}\tau(g_1^1)\otimes\tau(g_1^2)\right]\\
     & \quad \cdot\left(\prod_{i=2}^{m-1}\lbra{V\otimes V}\tau(g_i^1)\otimes\tau(g_i^2)\sum_{\sigma\in \mathbb P_n^0}\Tr(W \sigma W\sigma)\lket{\sigma\otimes\sigma}\right).
     \end{split}
     \end{equation}
   For an arbitary normalized Pauli operator $\sigma$, there is $g(\sigma)=g^{\dagger}\sigma g, g^{-1}(\sigma)=g\sigma g^{\dagger}$. We have that $\lbra{\sigma}\tau(g)=\lbra{g(\sigma)}$, thus,
    \begin{equation}\label{eq:b7}
    \begin{split}   \tau(g_m^1)\otimes\tau(g_m^2)\sum_{\sigma\in \mathbb P_n^0}\Tr(W \sigma W \sigma)\lket{\sigma\otimes\sigma}=&\sum_{\sigma\in \mathbb P_n^0} \Tr(W \sigma W\sigma)\lket{g^{1^{-1}}_m(\sigma)\otimes g^{2^{-1}}_m(\sigma)},\\
    \lbra{V^{\dagger}\otimes V}\tau(g_i^1)\otimes\tau(g_i^2)=&\lbra{g_i^1(V)\otimes g_i^2(V)}.
    \end{split}
    \end{equation}
    Denote $D=(d^2-1)^{2(m-1)},\ \delta_{i}=\delta_{g_i^1(V),g_i^2(V)}$. Combine Eq.~\eqref{eq:b7} and Eq.~\eqref{eq:b6}, then we can obtain that
    \begin{equation}\label{eq:f_A2}
    \begin{split}
    f_A(x,y,\mathbf{g^1},\mathbf{g^2}, m)=&D \Tr\left(B_{xy}\sum_{\sigma\in \mathbb P_n^0}\Tr(W \sigma W\sigma)\lket{g^{1^{-1}}_m(\sigma)\otimes g^{2^{-1}}_m(\sigma)}\lbra{g_1^1(V)\otimes g_1^2(V)}\right)\\
    &\quad\quad\cdot\prod_{i=2}^{m-1}\lbra{g_i^1(V)\otimes g_i^2(V)}\sum_{\sigma\in \mathbb P_n^0}\Tr(W \sigma W\sigma)\lket{\sigma\otimes\sigma}\\
    =&D \Tr\left(B_{xy}\sum_{\sigma\in \mathbb P_n^0}\Tr(W \sigma W\sigma)\lket{g^{1^{-1}}_m(\sigma)\otimes g^{2^{-1}}_m(\sigma)}\lbra{g_1^1(V)\otimes g_1^2(V)}\right)\prod_{i=2}^{m-1}\delta_{g_i^1(V),g_i^2(V)}\Tr(W g_i^1(V))^2\\
     =&D\left(\prod_{i=2}^{m-1}\delta_{i}\Tr(W g_i^1(V))^2\right)\Tr\left(\lket{\rho}\lbra{E_x}\otimes\lket{\rho}\lbra{E_y}\sum_{\sigma\in \mathbb P_n^0}\Tr(W \sigma W\sigma)\lket{g^{1^{-1}}_m(\sigma)\otimes g^{2^{-1}}_m(\sigma)}\lbra{g_1^1(V)\otimes g_1^2(V)}\right)\\
     =&D\left(\prod_{i=2}^{m-1}\delta_{i}\Tr(W g_i^1(V))^2\right)\lbra{E_x\otimes E_y}\sum_{\sigma\in \mathbb P_n^0}\Tr(W \sigma W\sigma)\lket{g^{1^{-1}}_m(\sigma)\otimes g^{2^{-1}}_m(\sigma)}\lbra{g_1^1(V)\otimes g_1^2(V)}\lket{\rho\otimes\rho}\\
     =&D\left(\prod_{i=2}^{m-1}\delta_{i}\Tr(W g_i^1(V))^2\right)\lbra{g_m^1(E_x)\otimes g_m^2(E_y)}\sum_{\sigma\in \mathbb P_n^0}\Tr(W \sigma W\sigma)\lket{\sigma\otimes \sigma}\lbra{g_1^1(V)\otimes g_1^2(V)}\lket{\rho\otimes\rho}\\
     =&D\left(\prod_{i=2}^{m-1}\delta_{i}\Tr(W g_i^1(V))^2\right)\lbra{g_m^1(E_x)\otimes g_m^2(E_y)}(\lket{S(W^{\dagger}\otimes W)}-\frac{1}{d}\lket{I\otimes I})\lbra{g_1^1(V)\otimes g_1^2(V)}\lket{\rho\otimes\rho}\\
     =&D\left(\prod_{i=2}^{m-1}\delta_{i}\Tr(W g_i^1(V))^2\right)\left(\Tr(W^{\dagger}g_m^1(E_x)Wg_m^2(E_y))-\frac{1}{d}\right) \Tr(g_1^1(V)\rho)\Tr(g_1^2(V)\rho).
\end{split}
\end{equation}
With the equation above, one can quickly evaluate the value of $f_A(x, y, \mathbf{g^1}, \mathbf{g^2}, m)$.

%



\end{document}